\documentclass[fleqn,11pt]{article}

\usepackage{graphicx}
\usepackage{color}
\usepackage{amsfonts, amsmath, amsthm, amssymb,epstopdf}
\usepackage{comment}
\usepackage{sectsty}
\sectionfont{\fontsize{12}{15}\selectfont}
\usepackage[margin=2cm]{geometry}
\usepackage[round]{natbib}
\usepackage{setspace}
\usepackage{placeins}
\usepackage{subfigure}

\usepackage[para,online,flushleft]{threeparttable}

\title{\textbf{The Classification Permutation Test: A Nonparametric Test for Equality of Multivariate Distributions}
}

\author{ \\  \Large 
Johann Gagnon-Bartsch\thanks{Johann Gagnon-Bartsch: 
Department of Statistics, University of Michigan,  Ann Arbor, MI, USA; \texttt{johanngb@umich.edu}.}
 \qquad \quad \quad 
Yotam Shem-Tov\thanks{Yotam Shem-Tov: 
Department  of  Economics,
University  of  California Berkeley,  530  Evans  Hall,  Berkeley,  CA  94720,  USA; \texttt{shemtov@berkeley.edu}.
\newline
\indent
\hspace{1.5mm}
We thank Isaiah Andrews, David Card, Avi Feller, Brian Graham, Patrick Kline, Edward Miguel, Demian Pouzo, Jasjeet Sekhon, Juliana Londo\~{n}o V\'{e}lez, and participants at the UC Berkeley Development Lunch Seminar and the Political Economy Lunch Seminar for helpful discussions and comments. 
}}

\date{\today}

\newcommand{\indep}{\rotatebox[origin=c]{90}{$\models$}}
\newcommand{\thth}{^{\mathrm{th}}}
\newcommand{\given}{ \, \middle| \, }
\newtheorem{lemma}{Lemma}
\newcommand{\Rsymb}{{\bf{\textsf{R} }}}

\begin{document}
\doublespacing

\newtheorem{assumption}{Assumption}
\newtheorem{theorem}{Theorem}
\newtheorem{proposition}{Proposition}
\newtheorem{corollary}{Corollary}
\newtheorem{definition}{Definition}

\begin{titlepage}
 \maketitle
  \begin{abstract}
\onehalfspacing

The gold standard for identifying causal relationships is a randomized controlled experiment.
In many applications in the social sciences and medicine, the researcher does not control the assignment mechanism and instead may rely upon natural experiments, 
regression discontinuity designs, RCTs with attrition, or matching methods as a substitute to experimental  randomization. 
The standard testable implication of random assignment is covariate balance between 
the treated and control units.  Covariate balance is therefore commonly used to validate the claim of ``as-if'' random assignment.  We develop a new nonparametric test of covariate balance.  Our Classification Permutation Test (CPT) is based on a combination of classification methods (e.g. logistic regression or random forests) with Fisherian permutation inference.  The CPT is guaranteed to have correct coverage and is consistent under weak assumptions on the chosen classifier. To illustrate the gains of using the CPT, we revisit four real data examples: \cite{lyall2009}; \cite{green2010}; \cite{eggers2009}; and \cite{rouse1995}.  Monte Carlo power simulations are used to compare the CPT to two existing nonparametric tests of equality of multivariate distributions. \\~\\

%	C12 	Hypothesis Testing: General
%	C18 	Methodological Issues: General 
% 	P16 	Political Economy
%  	K14 	Criminal Law 
\end{abstract}
\end{titlepage}

%%%%%%%%%%%%%%%%%%%%%%%%%%%%%%%%%%%%%%%%
%%%%%%%%%%%%%%  Intro  %%%%%%%%%%%%%%%%%
%%%%%%%%%%%%%%%%%%%%%%%%%%%%%%%%%%%%%%%%

\section{Introduction}  \label{introduction}

Many applications in the social sciences, economics, biostatistics, and medicine argue for ``as-if'' random assignment of units to treatment regimes.  Examples include natural experiments, regression discontinuity designs, matching designs, and RCTs with attrition.  To support a claim of ``as-if'' random assignment, researchers typically demonstrate that the observed covariates are balanced between treatment and control units. Typically it is required to show that pre-treatment characteristics cannot predict future treatment status.  

This paper develops a nonparametric test that formalizes the question of whether the covariates can predict treatment status.  The test makes use of classification methods and permutation inference, and we name it the Classification Permutation Test (CPT).  The CPT trains a classifier (e.g., logistic regression, random forests) to distinguish treated units from control units.  Then, using permutation inference, the CPT tests whether the classifier is in fact able to distinguish treated units from control units more accurately than would be expected by chance. 

The CPT may be viewed as a test for equality of multivariate distributions, as it tests whether the joint distribution of the covariates is the same in both the treatment and control groups.  Several other nonparametric tests for equality of multivariate distributions have been proposed in the past.  \cite{rosenbaum2005} developed the Cross-Match test which compares two multivariate distributions using a matching algorithm.  First, the observations are matched into pairs, using a distance metric computed from the covariates (treatment status is ignored). The Cross-Match test statistic is then the number of matched pairs containing one observation from the treatment group and one from the control group; high values of the test statistic imply covariate balance, and for low values the null hypothesis of random assignment is rejected.  Applications and extensions of the Cross-Match test are described in \cite{heller2010} and \cite{heller2010b}.  \cite{szekely2009, szekely2009b} developed the energy test, another nonparametric test for equality of multivariate distributions. \cite{aronow2012} suggested using the energy test to test for covariate imbalance between groups. \cite{cattaneo2014} proposed a permutation based method for optimal window selection in a regression discontinuity design based on covariate balance on both sides of the cut-point. The method uses only information about the marginal distributions of the covariates, and therefore may not detect imbalances in the joint distribution.  Still other methods include \citet{heller2013} and \citet{taskinen2005}. 

This paper contributes to the existing literature in four ways:
(1) We show that the CPT is a useful tool in practice.  Using both simulated and real data, we find that the CPT is often able to detect covariate imbalance where existing nonparametric methods do not.  
(2) The paper illustrates how ``black box'' algorithms from the machine learning literature such as random forests can be used for rigorous inference in the social sciences, without actually relying on any strong modeling assumptions.  Classification methods and permutation inference have been previously combined in the computational biology literature \citep{ojala2010}.  
(3) We apply the CPT to make a substantive contribution to the political economy and criminal justice literatures. We revisit \cite{eggers2009} and shed new light on the validity of their regression discontinuity design, and provide new evidence in support of the ``judges design'' identification strategy used by \cite{green2010}.         
(4) The CPT has a clear and intuitive interpretation.  The test statistic is a direct measure of the ability of the covariates to predict treatment assignment.  Moreover, the CPT relates equality of multivariate distributions to the propensity score \citep{rosenbaum1983}.  Rejection of the null hypothesis implies the covariates are predictive of treatment assignment, or in other words that the distribution of the propensity score is different across the treatment and control groups.

The paper is organized as follows. Section \ref{method} provides a brief overview of the method.  Section \ref{simulations} examines the performance of the CPT on simulated data, and Section \ref{applications} looks at real-life data examples.  Section \ref{technical} provides further theoretical discussion, including a proof that the CPT is consistent under weak assumptions on the chosen classifier.  

%%%%%%%%%%%%%%%%%%%%%%%%%%%%%%%%%%%%%%%%
%%%%%%%%%%%%%%  Overview  %%%%%%%%%%%%%%
%%%%%%%%%%%%%%%%%%%%%%%%%%%%%%%%%%%%%%%%

\section{Overview of the Method}  \label{method}

This section gives an informal description of the CPT and a more detailed description is given in Section \ref{technical}.
Suppose there are $n$ units, indexed by $i$.  For each unit there is a vector of observed covariates $Z_i \in \mathbb{R}^m$ and a treatment assignment indicator $T_i \in \{0,1\}$.  (Presumably there is an outcome variable as well, but it is irrelevant for our purposes.) We model the $(Z_i, T_i)$ pairs as being IID from some unknown distribution.  Let $T$ be the $n\times 1$ vector whose $i\thth$ entry is $T_i$ and let $Z$ be the $n\times p$ matrix whose $i\thth$ row is $Z_i$.
We wish to test whether 
\begin{equation}
T \indep Z \label{tindz}
\end{equation} 
or whether treatment assignment is independent of the observed covariates.  This is our notion of ``random treatment assignment.''  

The CPT proceeds as follows.  First, we train a classifier to predict $T$ from $Z$.  The classifier can be anything --- logistic regression, a random forest, K-nearest neighbors, etc.  We only require that the classifier provide us with a $n\times 1$ vector $\hat{T}$ of ``predicted'' treatment assignments, where $\hat{T}_i \in \{0,1\}$.  We then define the \textit{in-sample classification accuracy rate} $S$ as 
\begin{equation}
S \equiv \frac{1}{n}\sum_{i=1}^nI\{\hat{T}_i = T_i\}
\end{equation}
where $I\{\hat{T}_i = T_i\}$ is the indicator function for whether $\hat{T}_i = T_i$.  We use $S$ as our test statistic; intuitively, $S$ should be high only if $Z$ is predictive of $T$, implying that $Z$ and $T$ are not independent.  

To determine statistical significance, we use permutation inference.  We randomly permute the rows of $T$ (but not $Z$) $B$ times.  Each time we retrain the classifier and recalculate the classification accuracy rate, which we denote $S^{\star}_b$, where $1 \le b \le B$.  We then calculate our $P$-value as
\begin{equation}
\frac{1}{B}\sum_{b=1}^BI\{S \ge S^{\star}_b\}
\end{equation}
where $I\{S \ge S^{\star}_b\}$ is the indicator function for whether $S \ge S^{\star}_b$.

A few comments: (1) Because we use permutation inference, the CPT's $P$-value is valid, even in finite samples, no matter what classifier we use.\footnote{Strictly speaking, this is true only as $B \to \infty$; for finite $B$, the distribtuion of $\{S^\star_b\}$ is only an approximation to the true permutation null distribution, and thus our $P$-value is only an approximation to the true permutation test $P$-value.} (2) In particular, the CPT's $P$-value is valid despite the fact that we use the in-sample classification accuracy rate.  Overfitting may occur, causing $S$ to be quite high, perhaps misleadingly so.  However, overfitting would cause the $S^\star_b$ to be high as well; thus, any overfitting problem is also manifested in the null distribution, and thereby effectively accounted for.  (3) The choice of classifier does affect the power of the test; the CPT will only have power if the classifier is able to distinguish the distribution of the covariates in the treatment group from the distribution of the covariates in the control group.  In this paper we focus on logistic regression (with all pairwise interaction terms included in the design matrix) and also random forests.  We select these classifiers because they are able to detect differences in the joint distribution of the covariates, as opposed to merely differences in the marginal distributions.  

In addition to the CPT as it is described above, we also consider some variants.  In one variant we replace the in-sample classification accuracy rate by an out-of-sample accuracy rate estimated by cross-validation.  This makes the CPT very computationally demanding, but gives it nice theoretical properties; see Section \ref{technical}.  In Section \ref{rouse} we consider a scenario in which the experimental units are blocked.  We implement a variant of the CPT in which we permute treatment assignment only within blocks.

%%%%%%%%%%%%%%%%%%%%%%%%%%%%%%%%%%%%%%%%
%%%%%%%%%%%%%%  Simulations  %%%%%%%%%%%
%%%%%%%%%%%%%%%%%%%%%%%%%%%%%%%%%%%%%%%%
\section{Simulations}  \label{simulations}

We use Monte Carlo simulations to study the power of the CPT, the Cross-Match test \citep{rosenbaum2005}, and the energy test \citep{szekely2009, szekely2009b}.  In each simulation we generate $n = 200$ observations; 100 are in treatment and 100 in control.  For each observation $i$ we generate a vector $Z_i$ of $m = 3$ covariates.  In the treatment group, the covariates are drawn from a multivariate normal distribution with mean 0 and variance $\Sigma_{\rho}$, where
\begin{equation}
\Sigma_{\rho} \equiv 
\left(\begin{array}{ccc}
1 & \rho & \rho \\
\rho & 1 & \rho \\
\rho & \rho & 1 \\
\end{array} \right).
\end{equation}
In the control group, the covariates are also drawn from a multivariate normal distribution with mean 0, but with variance $\Sigma_0 = I_{3\times 3}$.  In other words, the only difference in the distribution of the covariates between the treatment and control groups is the correlation.  In particular, the marginal distributions of the covariates are identical between the treatment and control groups.  Differences between treatment and control units cannot be detected using a balance table or a main effects regression.    

We vary the value of $\rho$ from 0 to 0.75 in increments of 0.05.  For each value of $\rho$ we generate 1,000 datasets as described above, and then run the CPT, Cross-Match test, and energy test on each dataset.  From this, we are able to approximate the power (at significance levels $\alpha = 0.05$ and $\alpha = 0.01$) of each test as a function of $\rho$.  Results are shown in Figure \ref{fig: power-test-main}.  In addition, a receiver operating characteristic (ROC) plot\footnote{See \cite{fawcett2006} for a description of ROC curves.} for $\rho = 0.5$ is shown in Figure \ref{fig: roc-test}.

Figure \ref{fig: power-test-main} shows the CPT has higher power for every level of $\rho$.  Figure \ref{fig: roc-test} shows that the CPT has a higher true positive rejection rate for every level of false rejection rate.  Together, Figures \ref{fig: power-test-main} and \ref{fig: roc-test} suggest the CPT typically outperforms the Cross-Match test and the Energy test with respect to power in this simulation. 

\begin{figure}[ht]
\centering
\caption{Power of the Cross-Match test, the energy test, and the CPT on simulated data.}
\label{fig: power-test-main}
\includegraphics[scale=0.38]{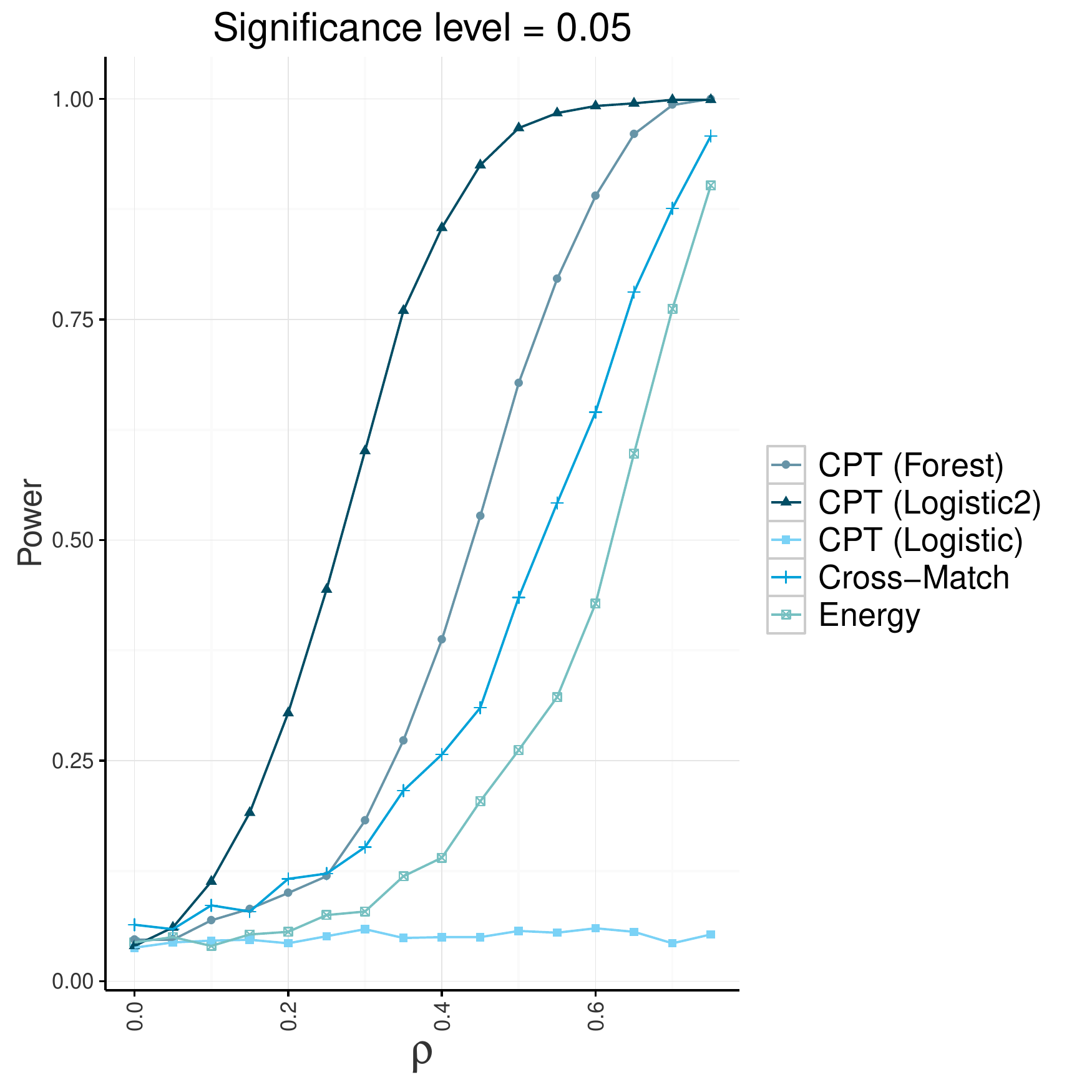} 
\includegraphics[scale=0.38]{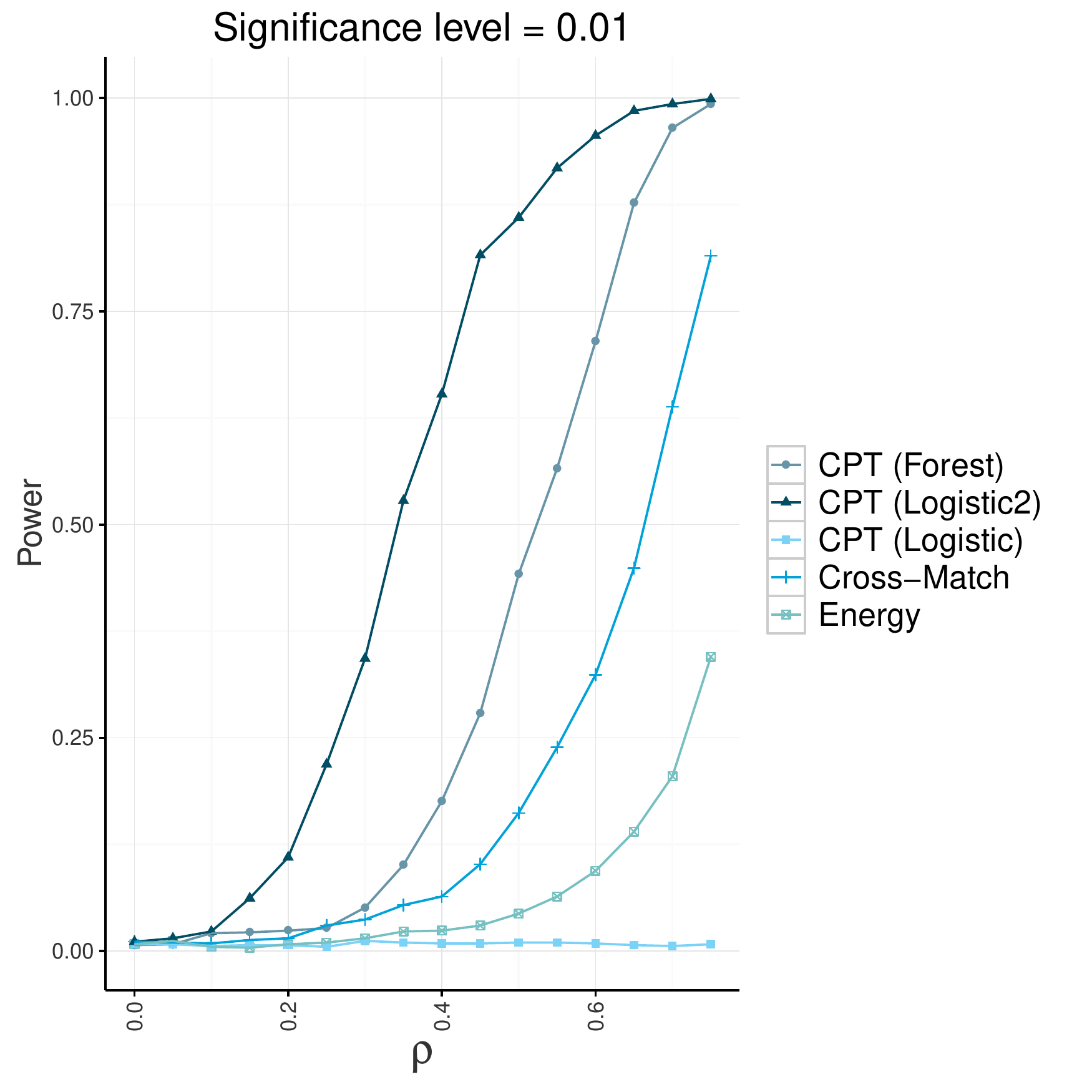}
\begin{minipage}{15cm}
\emph{Notes: Results for three variants of the CPT are shown; one variant (``logistic2'') uses a logistic regression classifier with all two-way interactions included in the model, and another (``forest'') uses random forests.  A third (``logistic'') uses logistic regression but does not include interaction terms; as expected, this version is unable to distinguish the two distributions.  We used $B = 500$ permutations in the calculation of the $P$-values.}
\end{minipage}
\end{figure} 

\begin{figure}[ht]
\centering
\caption{ROC curves (simulated data).}
\label{fig: roc-test}
\includegraphics[scale=0.38]{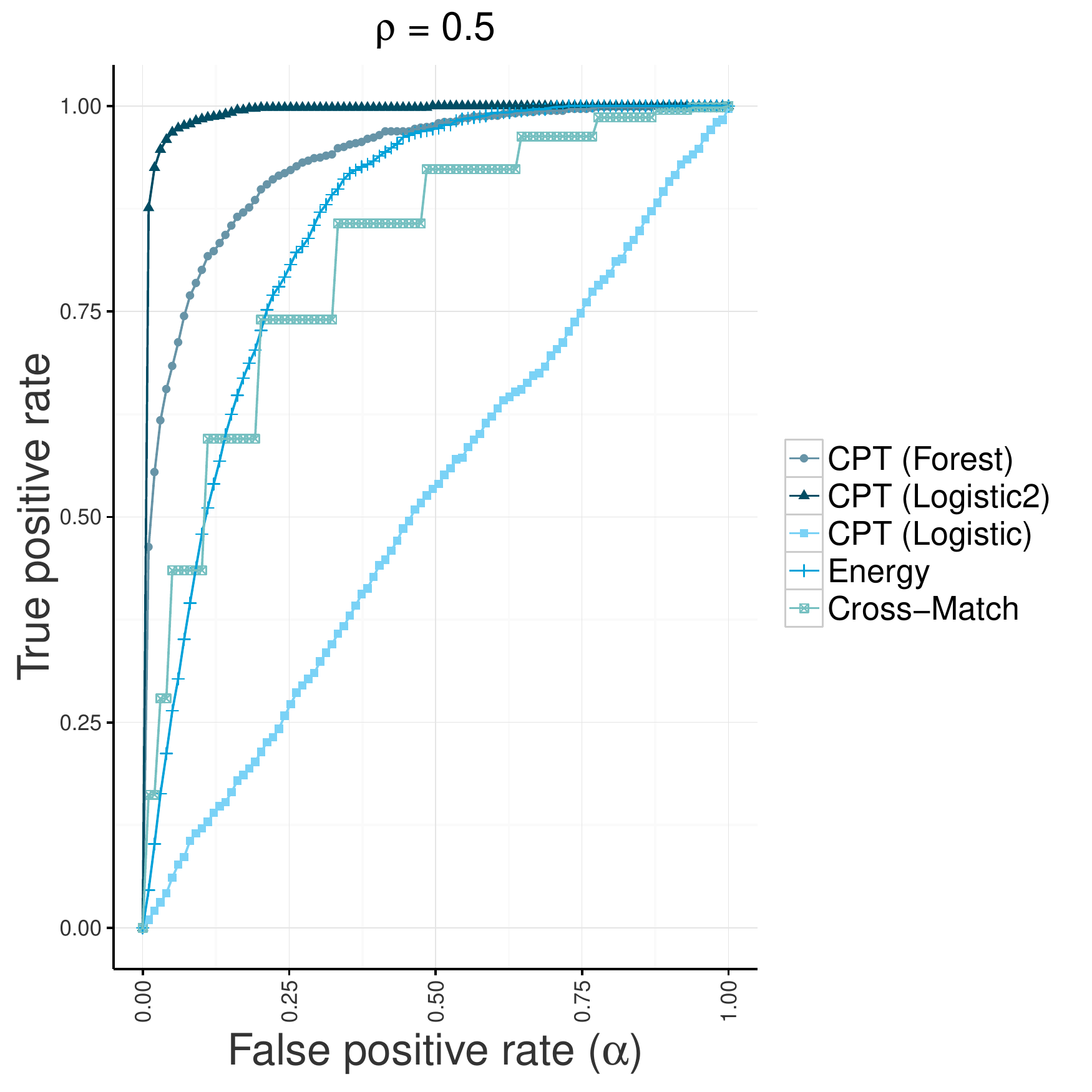}
\includegraphics[scale=0.38]{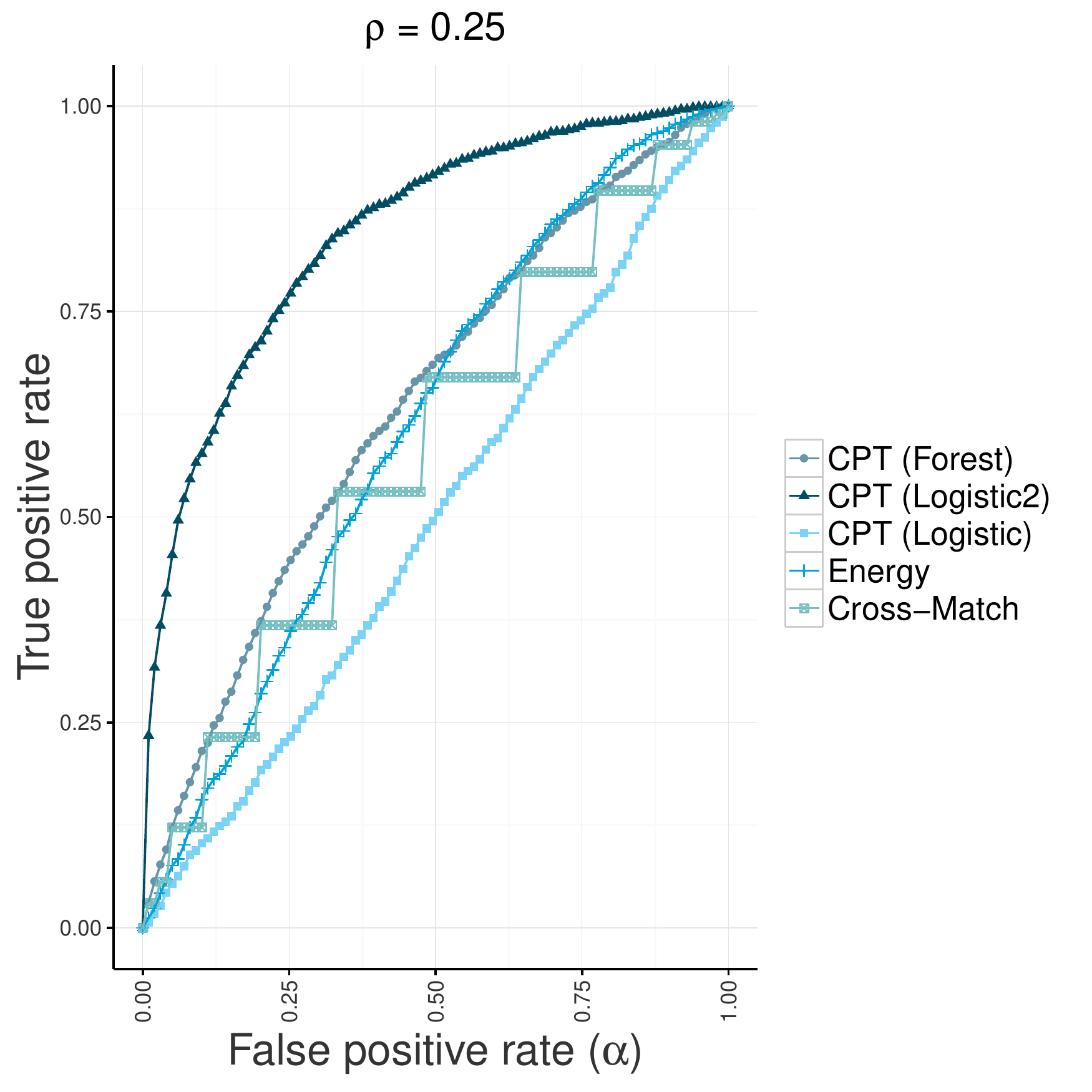}
\begin{minipage}{15cm}
\emph{Notes: See Figure \ref{fig: power-test-main}}.
\end{minipage}
\end{figure}

%%%%%%%%%%%%%%%%%%%%%%%%%%%%%%%%%%%%%%%%
%%%%%%%%%%%%%%  Applications  %%%%%%%%%%
%%%%%%%%%%%%%%%%%%%%%%%%%%%%%%%%%%%%%%%%

\clearpage
\section{Applications}  \label{applications}

%%%%%%%%%%%%%%%
%%% Lyall
%%%%%%%%%%%%%%%

\subsection{Indiscriminate Violence in Chechnya: \cite{lyall2009}}  \label{lyall}

\cite{lyall2009} investigates the effect of indiscriminate violence, specifically the bombing of villages in Chechnya.  Villages are the unit of analysis, and the outcome of interest is insurgent attacks. The identification strategy is a matching procedure that yields almost completely balanced treatment and control groups in all the marginal distributions.  See Figure \ref{fig: lyall-balance}.  Lyall also presents the results shown in Figure \ref{fig: lyall-balance}, and uses these results to support the claim of covariate balance.  

The CPT finds significant evidence of covariate imbalance between the treatment and the control groups. Figure \ref{fig: lyall-distro} shows the distribution of the CPT test statistic under the null and the observed test statistic. The null is clearly rejected. This example illustrates that looking only at the marginal distributions of the covariates is not sufficient.     

\begin{figure}[ht]
\caption{Covariate balance between treatment and control villages in \cite{lyall2009}.}
\label{fig: lyall-balance}
\centering
\includegraphics[scale=1]{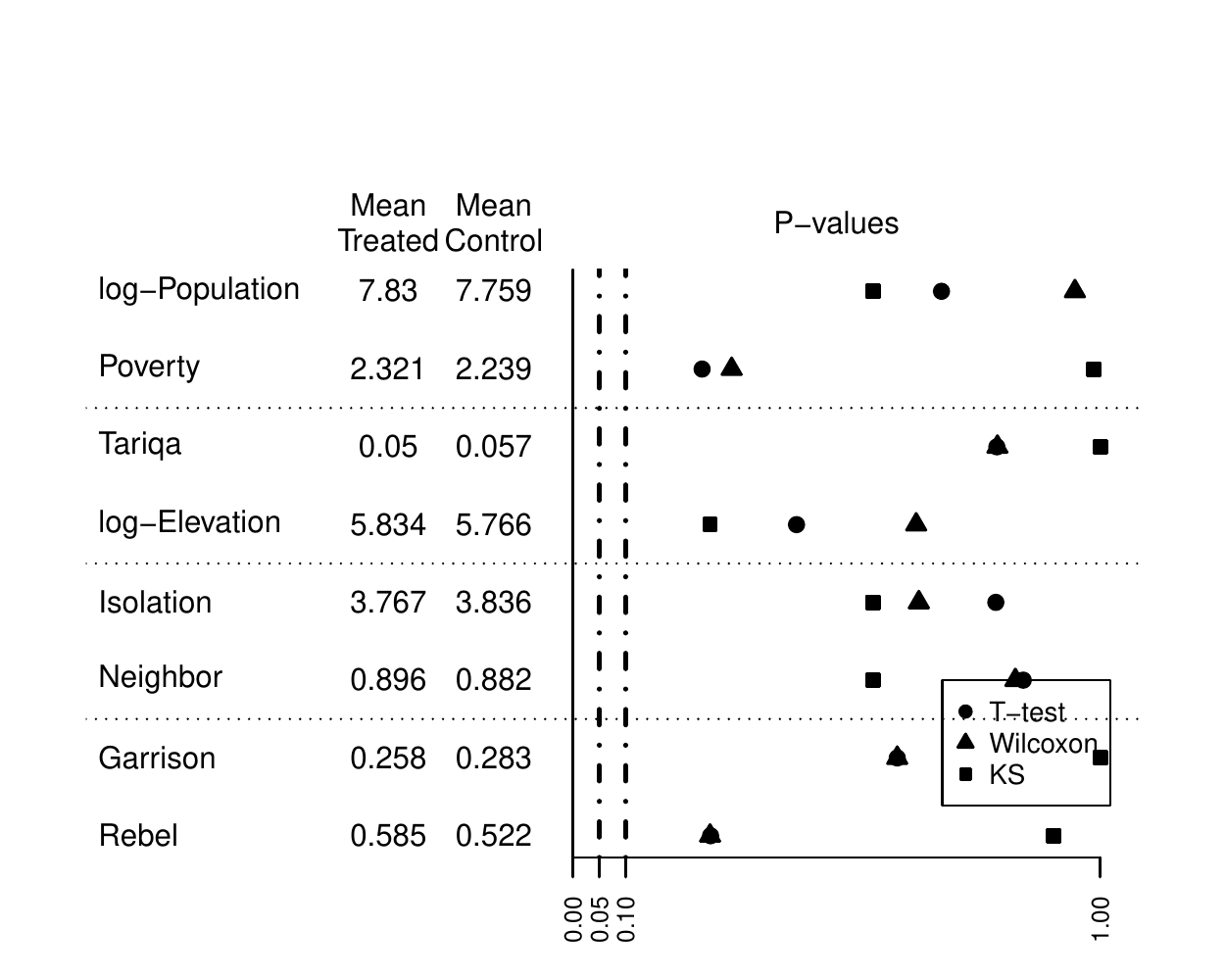}
\begin{minipage}{11cm}
\emph{Notes: Figure \ref{fig: lyall-balance} shows balance on each covariate separately. The points are P-values using t-test, Wilcoxon rank sum test and Kolmogorov-Smirnov (KS) tests.}    
\end{minipage}
\end{figure} 

\begin{figure}[ht]
\centering
\caption{Distribution of the CPT test statistic under the null hypothesis \newline} 
\label{fig: lyall-distro} 
\includegraphics[scale=0.48]{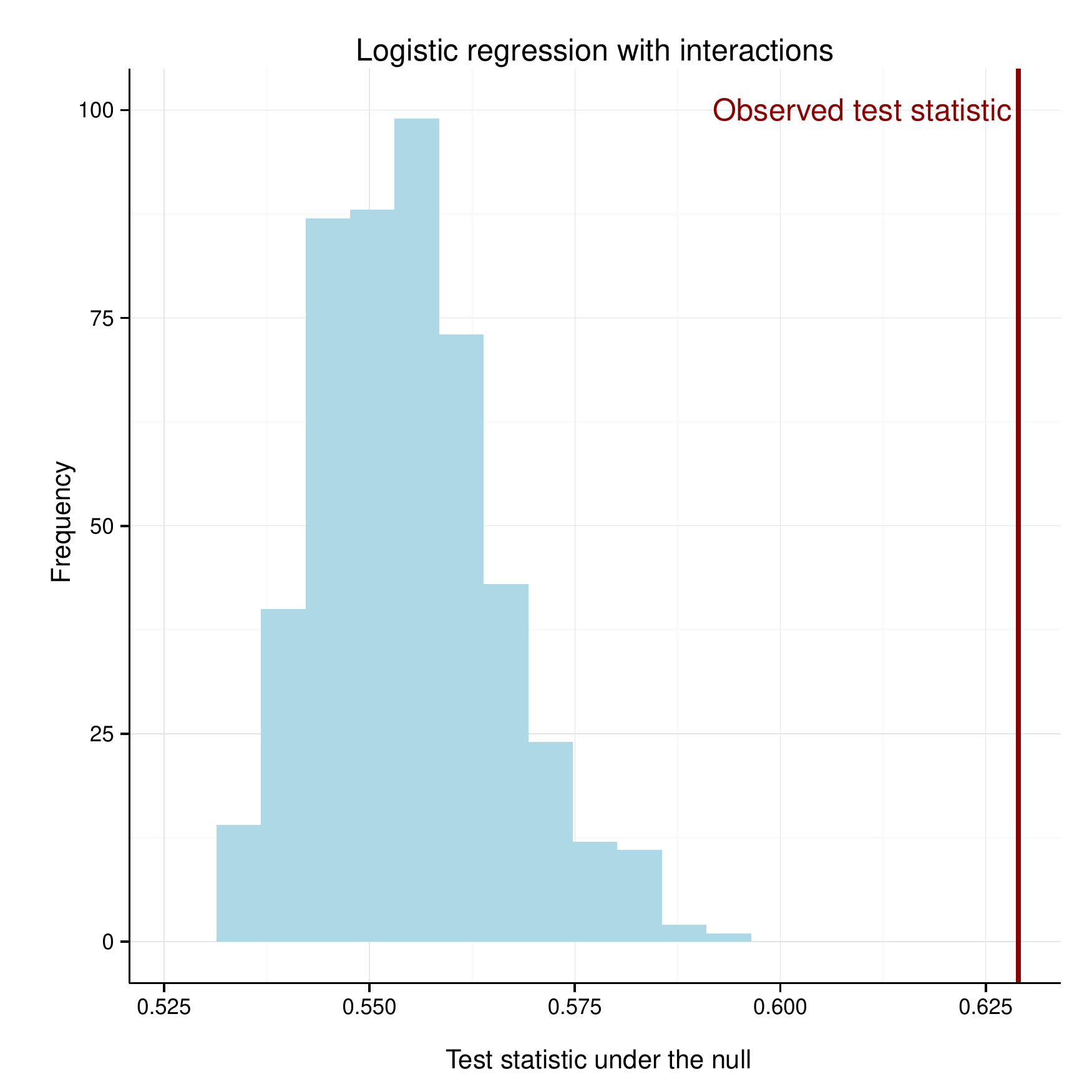} 
\includegraphics[scale=0.48]{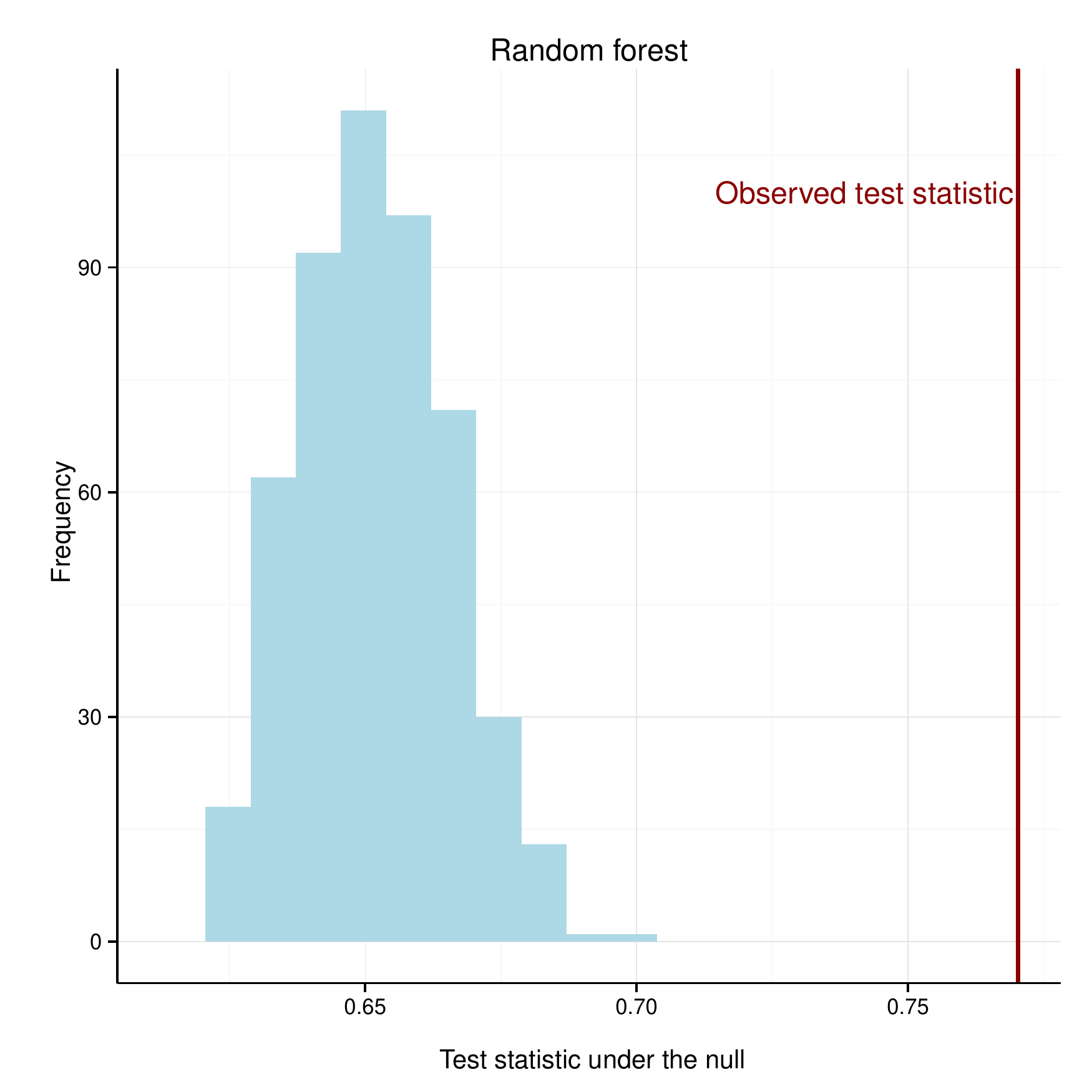}
\begin{minipage}{14cm}
\emph{\newline Notes: The figure shows the distribution of the CPT test statistic under the null hypothesis of random treatment assignment, and the observed test statistic. Results are shown for both logistic regression with all two-way interactions, and for random forests.}
\end{minipage}
\end{figure}

%%%%%%%%%%%%%%%
%%% Green
%%%%%%%%%%%%%%%

\subsection{Random assignment of defendants to judge calendars: \cite{green2010}} \label{judges}

\cite{green2010} studied the effect of incarceration length and probation length on recidivism. They argue that defendants are assigned as-if at random to different judge calendars, and that different judges have different punishment propensities. The data consists of a sample of $1,003$ felony drug defendants that are assumed to be randomly allocated between nine different judge calendars. The energy test ($P = 0.447 $) and the CPT both find no evidence of imbalance in the observed characteristics of the defendants across the nine judge calendars ($P = 0.115$). See Appendix \ref{appendix: data green and winik (2010)} for a list of all the observed covariates.  

\begin{figure}[ht]
\centering
\caption{The distribution of the estimated P-score using a main effects model and all two-way interactions model.   \newline}
\label{fig: distribution of P-score judge calendar 2}
\includegraphics[scale=0.45]{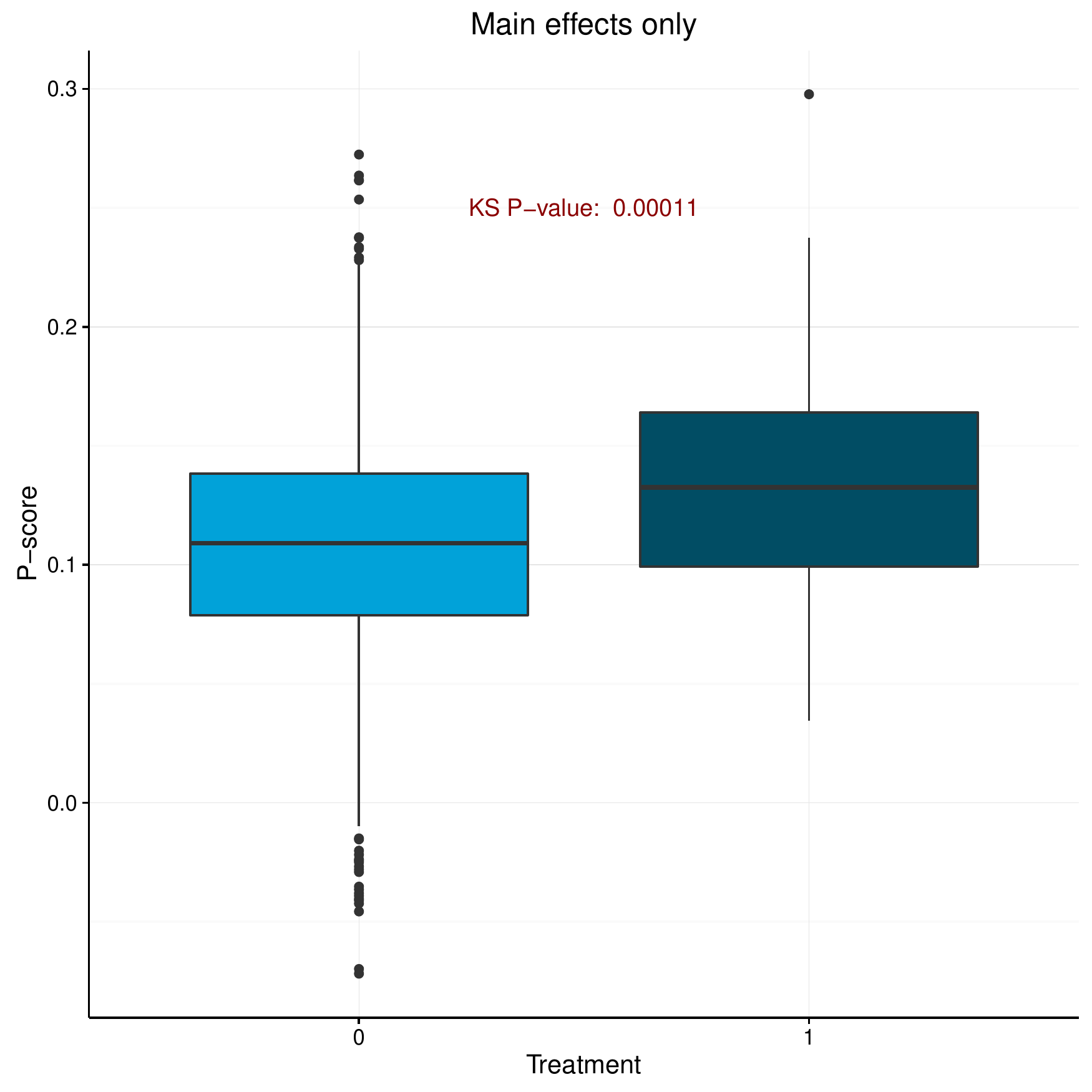} 
\includegraphics[scale=0.45]{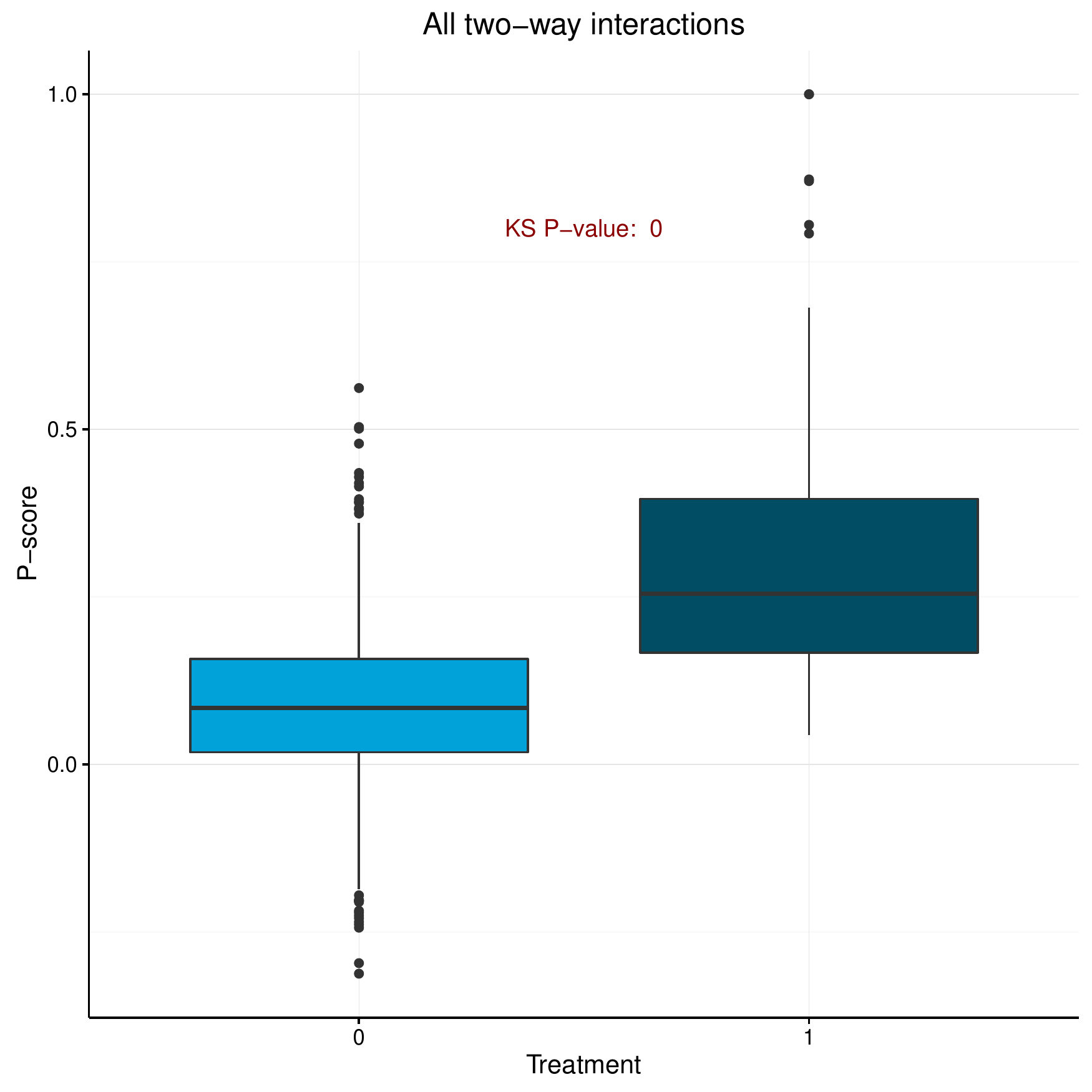} 
\end{figure}

An intuitive method for examining whether the observations in two groups are comparable in observable characteristics is to plot fitted propensity score ($e(Z)$) values, however this method can be sensitive to over-fitting issues.  Consider a binary indicator whether defendant $i$ was assigned to judge calendar $2$.\footnote{The choice of judge calendar $2$ is arbitrary and was motivated as an example that illustrates the issue of over-fitting a propensity score to the data. Other judge calendar choices are also possible, however our aim is not to deduce a statement on judge calendars, but rather to emphasize an estimation and testing issue.} We estimated $e(Z)$ using a logistic regression and plotted the fitted values, $\hat{e}(Z)$, among the treated (assigned to judge calendar $2$) and the control in Figure \ref{fig: distribution of P-score judge calendar 2}. The imbalance in the estimated propensity score could be the result of real differences in observable characteristics between the treated and control units or over-fitting of the logistic regression model to the observed data. The CPT does not find any difference in the observable characteristics between defendants assigned to judge calendar $2$ and the other defendants ($P = 0.241$). As the CPT re-estimates the logistic regression in each permutation it avoids any over-fitting issues and has finite sample exact coverage. 

The likelihood ratio test (LRT) from a logistic regression is a common alternative to the CPT or other permutation based tests. Table \ref{tab: type-I error rate judge calendar} shows the results of testing separately for each judge calendar whether defendants are randomly assigned or not using the LRT. The main effects logistic regression usually yields P-values that have correct coverage (i.e., Type-I error rate), however when including all two-way interactions the model over-fits the data and has incorrect coverage. This illustrates the over fitting problem of the LRT in finite samples. Next we investigate the finite sample performance of the LRT in this data application. 

Figure \ref{fig: judges Type-I error rate LRT} shows the distribution of the LRT P-values when the null hypothesis of random assignment is correct. We permuted the treatment at random and tested the null of random assignment. It is clear that the finite sample distribution of the over-fitted LRT P-value has incorrect Type-I error rates. The over-fitting problem of LRT in finite samples have been previously documented in the literature  \citep{hansen2008b}. 

\begin{table}[!htbp] \centering 
  \caption{The Likelihood Ratio Test P-values and Type-I error rates for each judge calendar dummy} 
  \label{tab: type-I error rate judge calendar} 
\begin{tabular}{@{\extracolsep{5pt}} ccccccc} 
\\[-1.8ex]\hline 
\hline \\[-1.8ex]
 & \multicolumn{3}{c}{\emph{Main effects only}} & 
\multicolumn{3}{c}{\emph{All two-way interactions}} \\ \\ 
 Judge calendar & P-value & Num. of coefficients & Type-I  & P-value & Num. of coefficients & Type-I  \\ 
\hline \\[-1.8ex] 
1 & $0.070$ & $22$ & $0.054$ & $0.001$ & $206$ & $0.354$ \\ 
2 & $0.210$ & $22$ & $0.069$ & $0.007$ & $206$ & $0.363$ \\ 
3 & $0.435$ & $22$ & $0.051$ & $0.123$ & $206$ & $0.343$ \\ 
4 & $0.852$ & $22$ & $0.063$ & $0.010$ & $206$ & $0.346$ \\ 
5 & $0.408$ & $22$ & $0.067$ & $0.129$ & $206$ & $0.339$ \\ 
6 & $0.767$ & $22$ & $0.067$ & $0.231$ & $206$ & $0.348$ \\ 
7 & $0.159$ & $22$ & $0.053$ & $0.017$ & $206$ & $0.354$ \\ 
8 & $0.917$ & $22$ & $0.066$ & $0.841$ & $206$ & $0.367$ \\ 
9 & $0.618$ & $22$ & $0.053$ & $0.090$ & $206$ & $0.363$ \\ 
\hline \\[-1.8ex] 
\end{tabular} 
\end{table}

%%%%%%%%%%%%%%%
%%% Eggers and Hainmueller
%%%%%%%%%%%%%%%

\subsection{MPs for Sale: \cite{eggers2009}}

\cite{eggers2009} (henceforth EH) studied the effect of membership in the UK parliament on personal wealth. EH use a regression discontinuity design (RDD) in which candidates for parliament who just barely won an election are compared to candidates who just barely lost.  In a RDD the observations just above and just below the threshold are assumed to be comparable, with the same distribution of observed and unobserved characteristics \citep{caughy2011}.  Testable implications of a valid RDD include covariate balance and no manipulation around the winning threshold \citep{imbens2008, lee2010}.

The aim of this data application is to illustrate the performance of the CPT in a RDD setting.  To  begin, we cast doubt on the RDD used by EH.  We demonstrate manipulation of the running variable (vote share) around the cut-point.  We also find imbalance in the party identity close to the cut point. 
Second, we drop party identity from the covariate set, and examine how well the CPT succeeds at identifying an imbalance in observables using only the remaining covariates. This can be thought of as a power test of how well the CPT can identify that the RD design is not valid. 

One possible explanation for our findings is that the EH design breaks the RD pairs of barely winners and losers by comparing individuals who attempted to run a \emph{different} number of times across \emph{multiple} elections. For example, the barely winners (and losers) could have run several times before the first winning (or best losing) race, and those elections will be ignored in the EH design. If for example the design used only one election at time X, this issue would have not been a problem. The concern raises from a comparison across multiple elections of candidates that are not necessary comparable due to differences in the characteristics that motivate a candidate to continue trying to be elected after losing a race. 
If the two populations of candidates, barely winners and losers, are indeed different in observable (and non-observable) characteristics it can explain our findings.

Figure \ref{fig: Density histograms MPs} shows the distribution of the winning margin by party. There is clear evidence of manipulation around the winning threshold by the non-labour party candidates. The McCrary test \citep{mccrary2008} for manipulation around the cut-point finds significant evidence of manipulation. 

\begin{figure}[ht]
\centering
\caption{The distribution of the winning margin by party identity \newline}
\label{fig: Density histograms MPs}
\includegraphics[scale=.3]{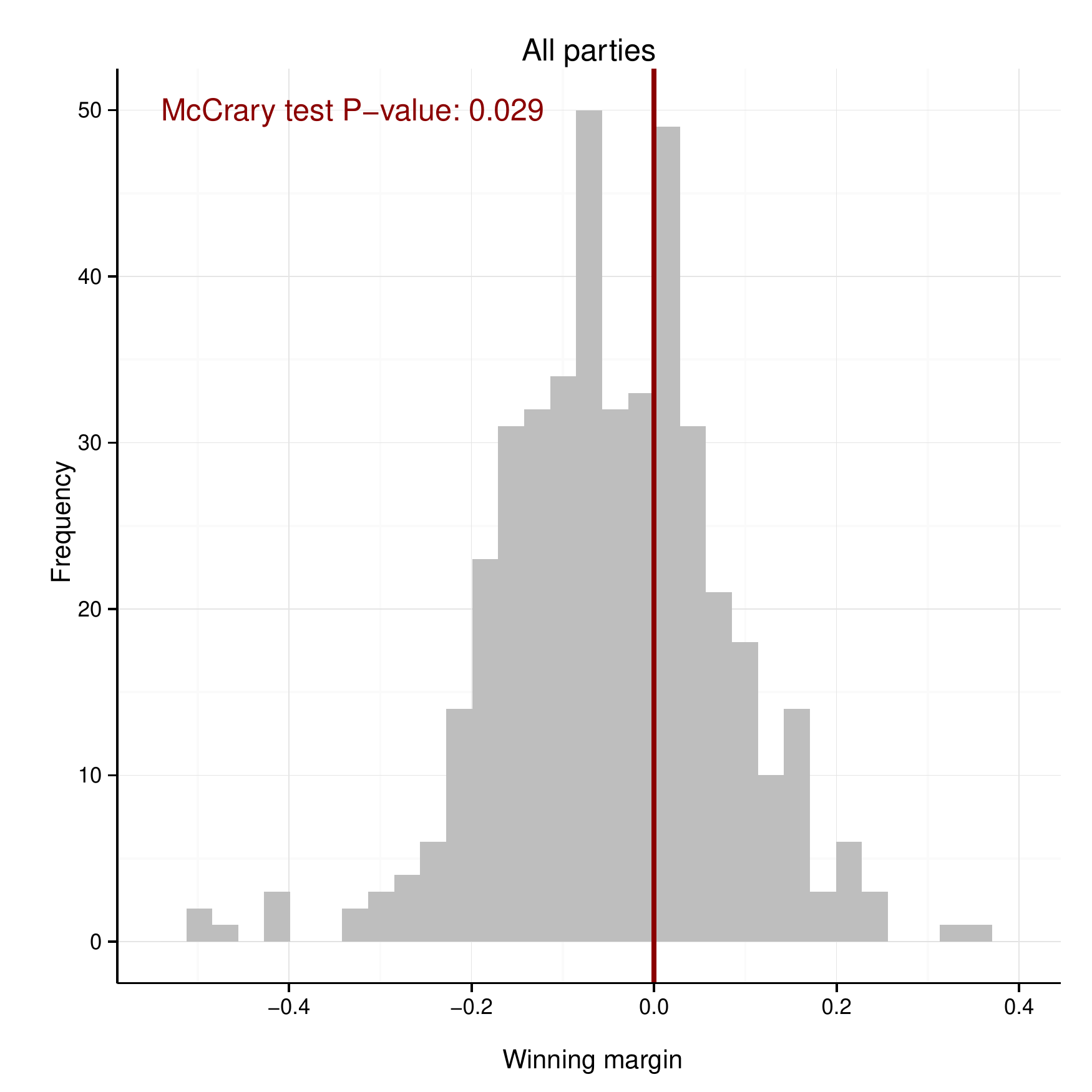}
\includegraphics[scale=.3]{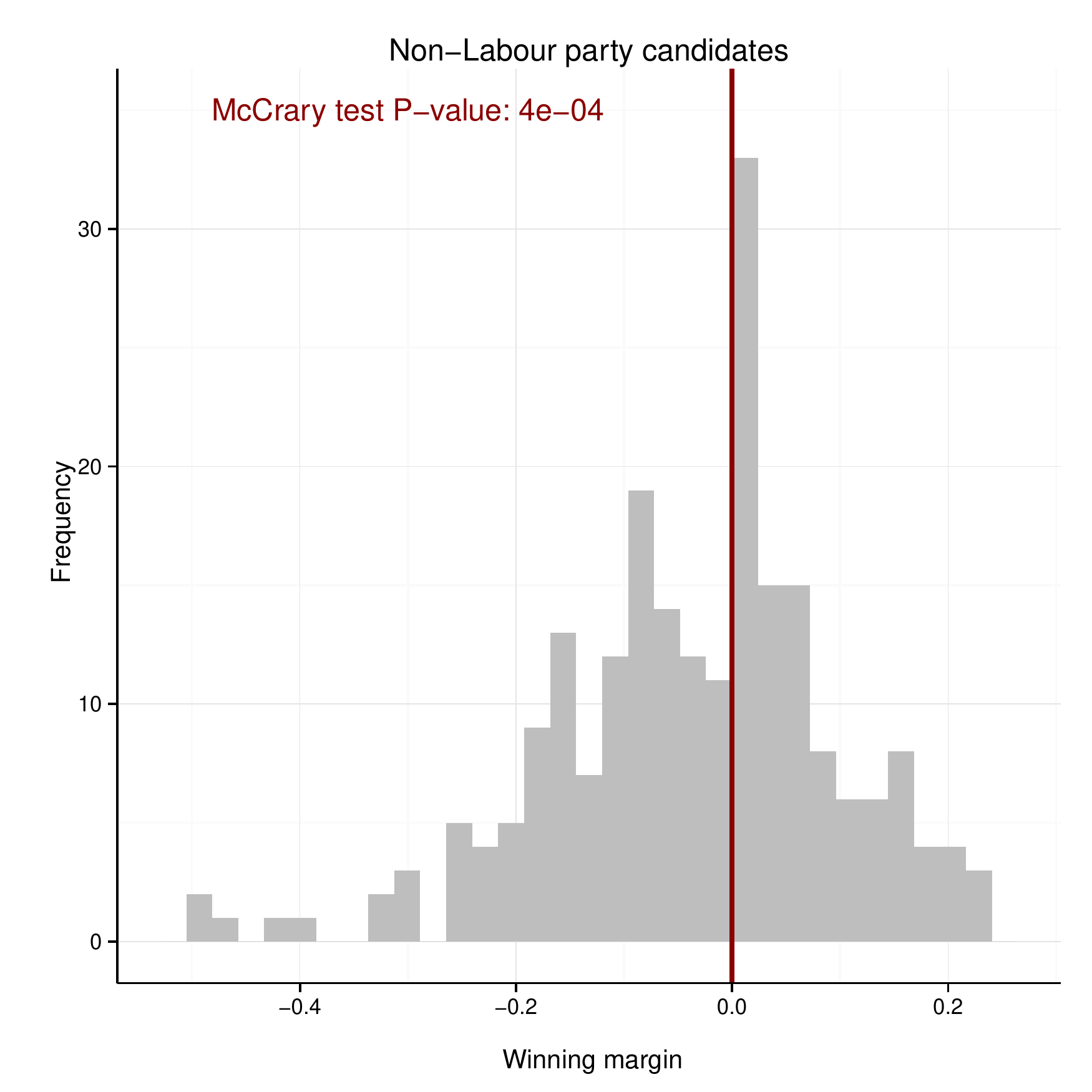}
\includegraphics[scale=.3]{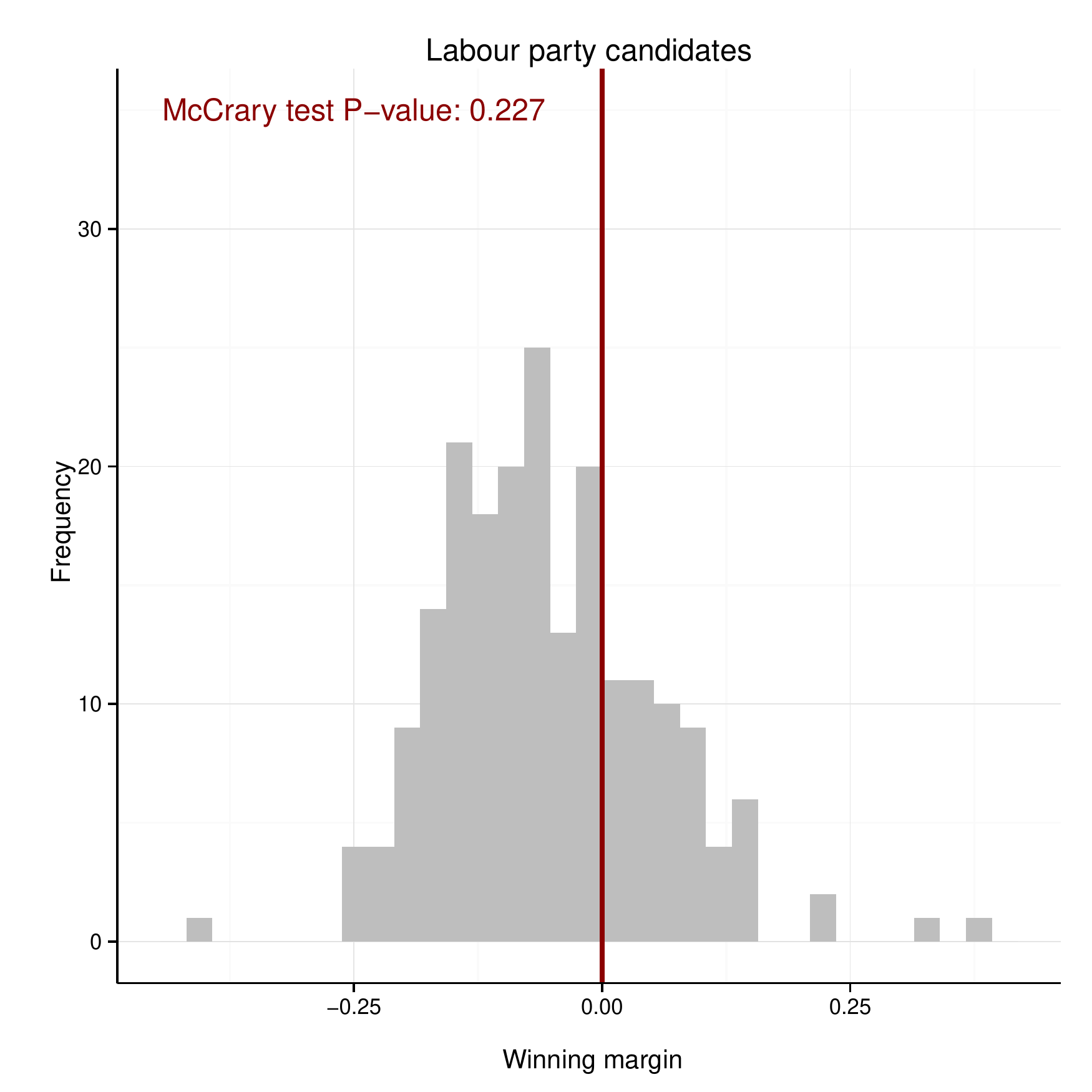}
\begin{minipage}{15cm}
\footnotesize
\emph{Notes: The bin size is at the default level in the \Rsymb package ``ggplot2''.}
\end{minipage}
\end{figure}

To demonstrate the added value of the CPT relative to a standard balance table, we will look at a specific window around the winning threshold. 
Table 4 in EH shows the main estimates of the treatment effect. The estimates use a window of 164 to 223 observations around the winning threshold. We restricted the sample to a window containing 164 observations, and examined the covariate balance within that window. Table \ref{tab: MPs balance window} and Figure \ref{fig: MPs balance table } in the Appendix suggest the covariate balance is not bad, and except for imbalance on the party identity, most other covariates seem to be balanced. Furthermore, a joint F-test of the null hypothesis that the covariates have no predictive power rejects only at a 10\% significance level ($P = 0.075$) and without including the party indicator the joint F-test does not reject the null of no predictive power ($P = 0.412$) and finds no evidence of imbalance when part identity is not included.      

We remove the party indicator from the covariate set and check whether the multivariate balance tests can detect a difference between the winners and the losers based on the remaining covariates.  
The Energy test and the Cross-Match test do not detect a covariate imbalance ($P=0.88$ and $P=0.13$ respectively), however the CPT finds significant imbalance between the two groups, see Figure \ref{fig: Null distribution CPT MPs} in the Appendix. 

In Figure \ref{fig: P-values for different window sizes} we compare the Energy test, Cross-Match test and the CPT over a grid of different window sizes. The results suggest that the CPT has higher power than the Energy and Cross-Natch tests. The CPT detects significant covariate imbalance at window sizes that are half of the one used by EH. We used a random forest as the classifier, because logistic regression with all two-way interactions had more parameters than observations. This is an example of how machine learning algorithms combined with permutation inference can be used to complement existing econometric tools.       

\begin{figure}[ht]
\centering
\caption{P-values of each of the multivariate balance test at different window sizes.   \newline}
\label{fig: P-values for different window sizes}
\includegraphics[width=0.78\textwidth]{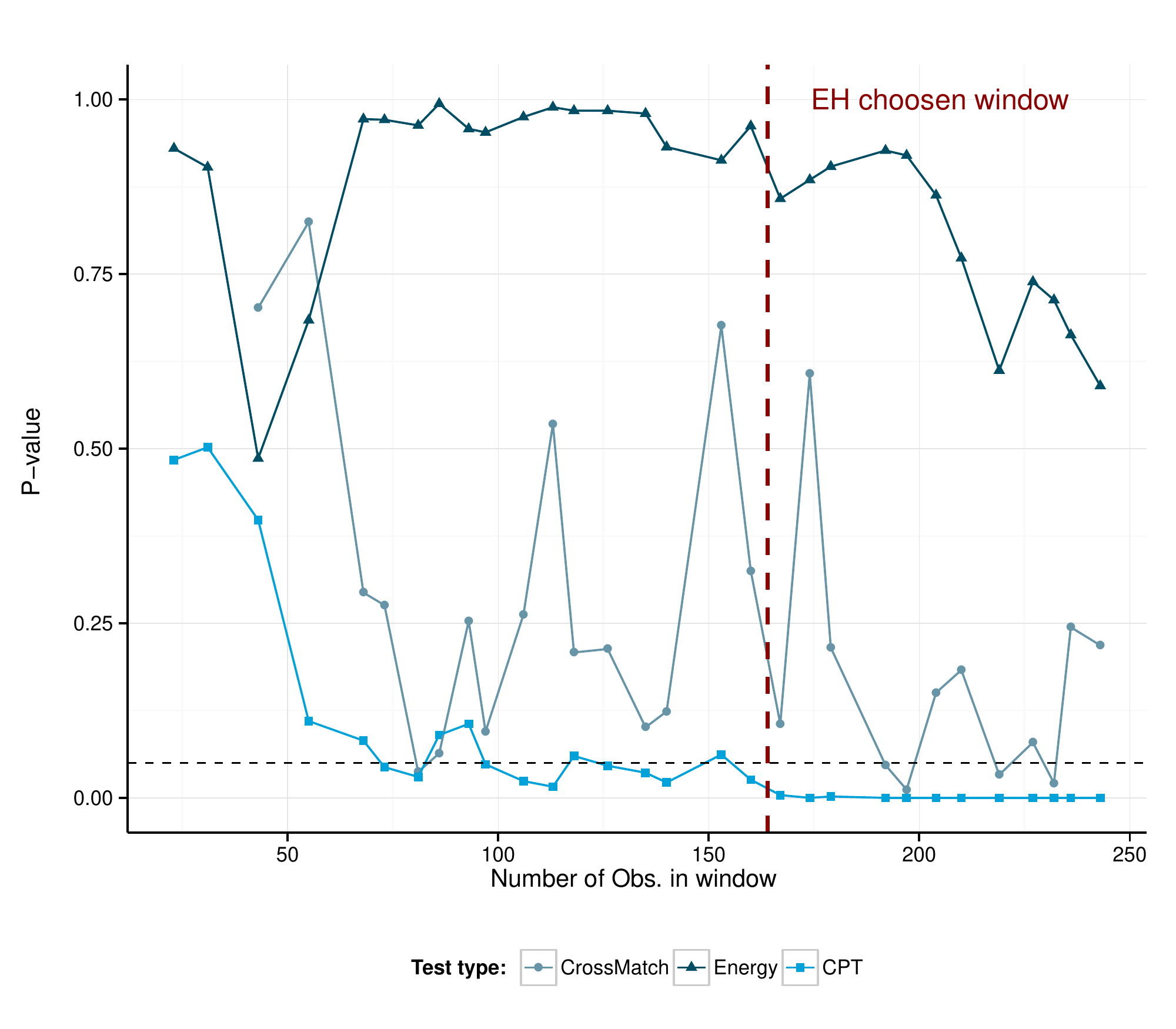}
\begin{minipage}{13.5cm}
\footnotesize
\emph{Notes: 
The CPT uses a random forest classifier, and the test statistic is the in-sample classification accuracy rate. 
In the smallest two window sizes the Cross-Match test statistic was not well defined, as the covariance matrix could not been inverted.
EH used a window containing between 164 to 223 observations in their RD treatment effect estimation, see Table 4 in EH. }
\end{minipage}
\end{figure}

%%%%%%%%%%%%%%%
%%% Rouse
%%%%%%%%%%%%%%%

\subsection{The effect of community college on educational attainment, \cite{rouse1995} and \cite{heller2010}} \label{rouse}

In a matching design it is common to use Fisherian inference after conducting the matching procedure, see \cite{rosenbaum2010}. A key question is whether after matching the researcher should imagine that units have been assigned at random within matched blocks, or whether each unit has been assigned independently to treatment.  In other words, in the hypothetical experiment that the matching design is meant to mimic, is the randomization within a group (match) or across groups?    
In this data application we will show it is essential to specify the probability model, because the two may lead to opposite conclusions when conducting balance diagnostics.      

\cite{rouse1995} studied the educational attainment of students who started in a two-year college to that of students at a four-year college. \cite{heller2010} used this data to demonstrate the use of the Cross-Match test for testing imbalance between multivariate distributions. We use this data to demonstrate methodological issues in conducting inference after matching, and not to make any inference or analysis on the effects of two-year college on educational attainment relative to four-year college.  

In Rouse's data, prior to conducting matching there is clear imbalance in the observable characteristics of students who started at a two-year college and those who started at a four-year college (see Figure \ref{fig: college balance} in Appendix). After matching, with or without replacement, the balance tables comparing the treated (two-year) and control (four-year) units show the groups are comparable in the observed characteristics and validates the matching procedure worked well. To test whether there is imbalance in the joint distribution of the covariates we use the CPT, and Figure \ref{fig: college null distribution CPT} shows the results. Figure \ref{fig: college null distribution CPT} yields opposite results depending on the randomization structure that is used. When the randomization structure is across blocks the observed test statistic is to the left of the null distribution, implying more balance than would have been likely under random assignment. When the randomization structure is within blocks the observed test statistic is to the right of the null distribution, implying the covariates can predict the treatment  assignment better than under random assignment. The difference between the left and right plots in Figure \ref{fig: college null distribution CPT} is the matching method, with or without replacement, and as can be seen the matching procedure has no effect on our discussion of within versus across block randomization. 
\begin{figure}[ht]
\begin{center}
\caption{The distribution of the test statistic under the null according to randomization within blocks and across blocks for matching designs with and without replacement}
\label{fig: college null distribution CPT}
\end{center}
\includegraphics[scale=0.5]{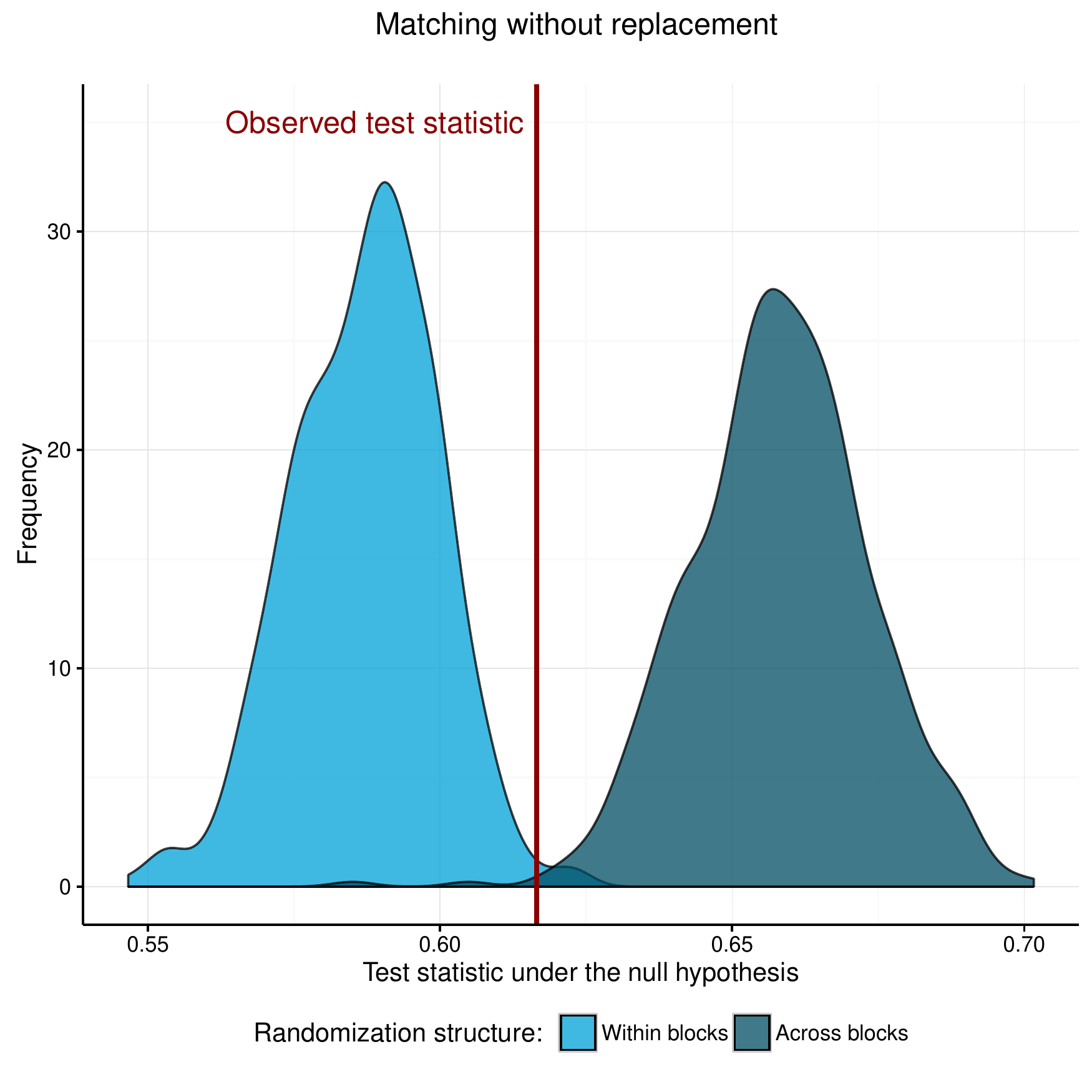}
\includegraphics[scale=0.5]{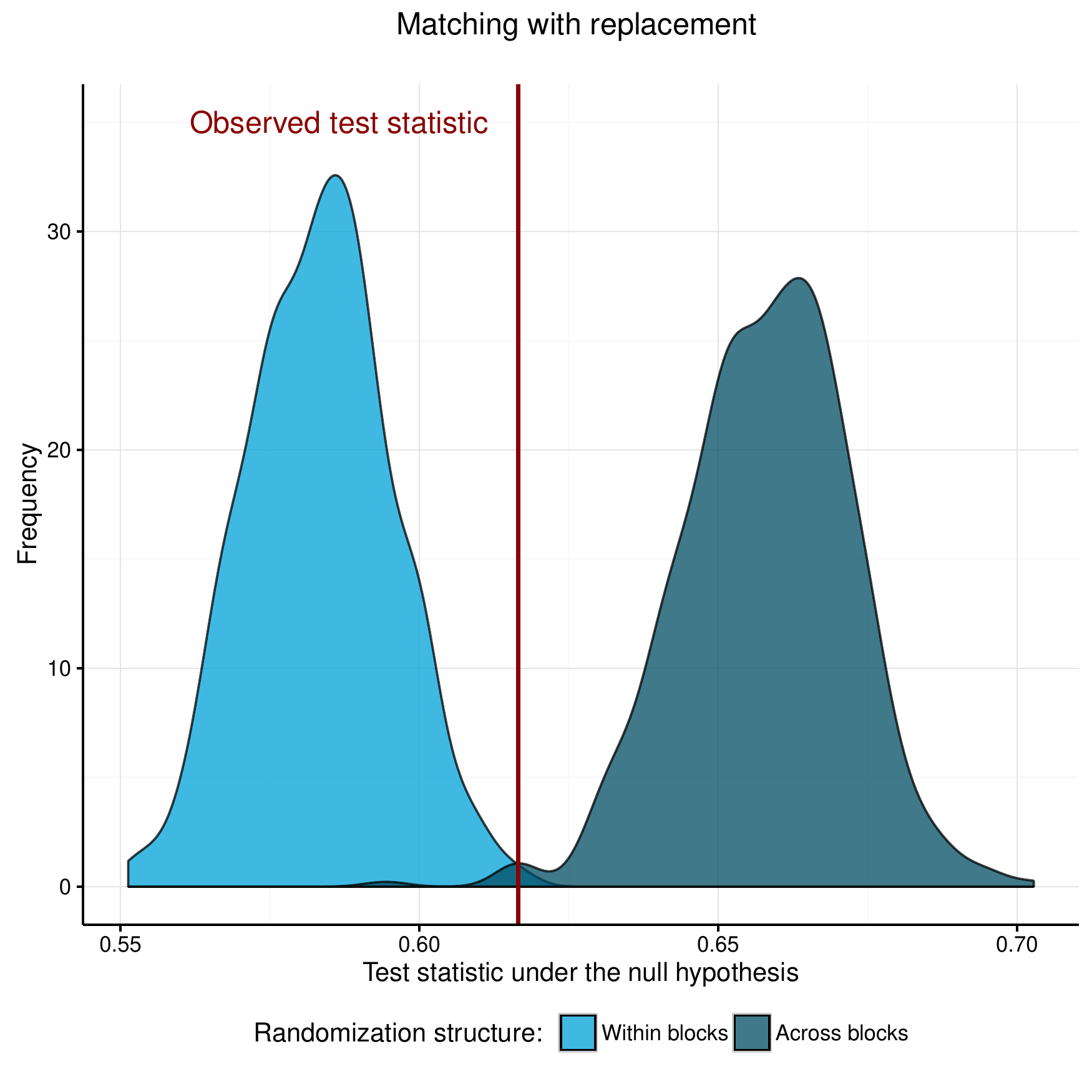}
\begin{minipage}{17cm}
\footnotesize
\emph{Notes:} The difference between the left and right panels is whether the matching was done with replacement or without, and as can be seen from the figure the matching method has no effect on our conclusions concerning within versus across block randomization.   
\end{minipage}
\end{figure}

%%%%%%%%%%%%%%%%%%%%%%%%%%%%%%%%%%%%%%%%
%%%%%%%%%%%%%%  Technical  %%%%%%%%%%%%%
%%%%%%%%%%%%%%%%%%%%%%%%%%%%%%%%%%%%%%%%

\section{Theory} \label{technical}

In this section we explicitly reformulate the CPT as a two-sample test for equality of multivariate distributions (\ref{reformulation}), describe an idealized version of the CPT (\ref{idealized}), and then show that the idealized CPT is consistent under weak conditions (\ref{consist}).  We conclude with some comments (\ref{techcomments}).

\subsection{Reformulation} \label{reformulation}

In Section \ref{method} we assume that the $(Z_i, T_i)$ pairs are IID from some unknown distribution.  Let $\mathcal{F}$ be the conditional distribution of $Z_i$ given $T_i = 1$ and $\mathcal{G}$ be the conditional distribution of $Z_i$ given $T_i = 0$.  Then $Z \indep T$ if and only if $\mathcal{F} = \mathcal{G}$.  We may therefore reformulate the CPT as a test for equality of the multivariate distributions $\mathcal{F}$ and $\mathcal{G}$.  

Suppose there are $l > 0$ values of $i$ for which $T_i = 1$, and let $X_1$, $X_2$, ..., $X_l$ denote the $l$ corresponding $Z_i$.  Similarly suppose there are $m > 0$ values of $i$ for which $T_i = 0$ and let $Y_1$, $Y_2$, ..., $Y_m$ denote the $m$ corresponding $Z_i$.  Let $X$ be the $l\times p$ matrix whose rows are $X_1$, ..., $X_l$, and let $Y$ be the $m \times p$ matrix whose rows are $Y_1$, ..., $Y_m$.  Note that the rows of $X$ are IID draws from $\mathcal{F}$ and the rows of $Y$ are IID draws from $\mathcal{G}$.  In this context, the CPT is simply a two-sample test comparing $X$ and $Y$. 

Let us now \textit{redefine} $Z$ to be the $n\times p$ matrix
\begin{equation}
Z \equiv \left(\begin{array}{c} X \\ Y \end{array} \right).
\end{equation}
Note that in our redefinition of $Z$ we have simply reordered the rows so that the first $l$ rows are from the treatment group and the remaining $m$ rows are from the control group.

\subsection{Description of an Idealized CPT} \label{idealized}

Let $s:\mathbb{R}^{n \times p}\mapsto \mathbb{R}$ be some fixed but otherwise arbitrary measurable function that maps an $n \times p$ matrix to a real number.  We will use $S \equiv s(Z)$ as our test statistic.  (We specify possible choices for $s$ below, but for now we allow $s$ to be arbitrary.)  

Let $\Pi_1$, ..., $\Pi_{n!}$ denote some ordering of the $n!$ permutation matrices of dimension $n \times n$.  We assume that $\Pi_1 = I$, but the ordering may otherwise be arbitrary.  Define $S^{(i)} \equiv s(\Pi_i Z)$ for $1 \le i \le n!$.  The values of $S^{(i)}$ are the re-calculated values of the test statistic we obtain after shuffling the observations (i.e., after shuffling the rows of $Z$).  

Now define
\begin{equation}
P \equiv \frac{\#\left\{i : S \le S^{(i)} \right\}}{n!}.
\end{equation}
Proposition \ref{validity} states that $P$ is a valid $P$-value for testing the hypothesis that $\mathcal{F}=\mathcal{G}$.
\begin{proposition} \label{validity}
Assume that $\mathcal{F} = \mathcal{G}$.  Then for any real number $\alpha$ such that $0 \le \alpha \le 1$, it follows that $\mathbb{P}\left(P \le \alpha \right) \le \alpha$.  
\end{proposition}
A proof is given in Appendix \ref{proofs}.  We must point out that there is of course nothing fundamentally new here; we have simply outlined a classical permutation test.  The key point we wish to make is that we may choose any function $s$ that we like, and the test remains valid.  Indeed, our choice of $s$ is a rather complicated function.  We use $Z$ to train a classifier that classifies observations as coming from either $\mathcal{F}$ or $\mathcal{G}$, and $s(Z)$ is some measure of the accuracy of the classifier.  (The function $s$ encapsulates both the training of the classifier, and the measurement of its accuracy.)  We also point out that we are describing here an idealized version of the CPT, because it is usually infeasible to compute $P$ in practice, since that would require us to compute all $n!$ values of $S^{(i)}$.  Lastly, we note that the assumption that $s$ is fixed is somewhat restrictive.  It excludes the possibility that the classifier might use a randomized algorithm.  This would exclude, for example, random forests.  In Appendix \ref{proofs} we discuss a generalization of Proposition \ref{validity} that allows for $s$ to be random.

We next discuss how we might construct our function $s$, and present two possibilities.  One possibility calculates the in-sample classification accuracy rate.  The other calculates the out-of-sample classification rate.  Both require us to specify a \textit{classification function}, which in practice amounts to choosing a classification algorithm (e.g. logistic regression).

\paragraph{Classification Function}

The classification function, which we denote $\hat{f}$, is a function that takes an observation and classifies it as coming from either $\mathcal{F}$ or $\mathcal{G}$.  Somewhat informally, we may think of $\hat{f}$ as a function that maps a $p$-dimensional vector (i.e., a single observation) to $\{0, 1\}$, with ``1'' meaning the observation is classified as coming from $\mathcal{F}$, and ``0'' meaning the observation is classified as coming from $\mathcal{G}$.  However, the prediction rule used by $\hat{f}$ to classify observations is learned from training data, and thus, strictly speaking, the function $\hat{f}$ is a function not only of the observation to be classified, but also of the training data.  We therefore write $\hat{f}$ as a function of two variables, i.e.\ $\hat{f}(u,v)$, where $u$ is the observation to be classified (a $p$-dimensional vector) and $v$ is the training data (a $n_0\times p$ matrix, where $n_0$ is the number of observations included in the training set; $n_0$ will usually be defined implicitly, depending on how we construct the training set.).  

In what follows, we allow $\hat{f}$ to be any fixed, measurable function that maps $\mathbb{R}^p\times\mathbb{R}^{n_0 \times p}$ to $\{0, 1\}$.  We do not place any other restrictions on $\hat{f}$.  In practice, we might choose $\hat{f}$ to be, for example, a logistic regression classifier.  Note that we require $\hat{f}$ here to be fixed, which excludes randomized algorithms such as random forests.  

\paragraph{In-sample Classification Accuracy Rate} Once we have chosen a function $\hat{f}$, we may then define $s$ in terms of $\hat{f}$.  One simple option is the \textit{in-sample classification accuracy rate}:
\begin{equation}
s_{\mathrm{in}}(z) = \frac{1}{n} \left\{\sum_{i = 1}^l \hat{f}(z_i, z) + \sum_{i = l+1}^n \left[1 - \hat{f}(z_i, z)\right]\right\}
\end{equation}
where $z_i$ denotes the $i\thth$ row of $z$ (note that the variable $z$ does not have any special meaning of its own; it is simply used here to define the function $s_{\mathrm{in}}$ in terms of the function $f$).  Here, we use the entire dataset $z$ as the training data (so $n_0 = n$).  We then count the number of observations that are correctly classified and divide by $n$.

\paragraph{Out-of-sample Classification Accuracy Rate} An alternative to the in-sample classification accuracy rate would be the \textit{out-of-sample classification accuracy rate}.   As defined below, the out-of-sample classification accuracy rate essentially amounts to cross validation except that we consider all possible training sets of a fixed size, instead of just 5 or 10 disjoint training sets.  In addition, we require that exactly half of the observations in the test set come from $\mathcal{F}$ and exactly half come from $\mathcal{G}$ (see discussion in Section \ref{techcomments}).

Let $\kappa$ be some integer such that $1 \le \kappa < \mathrm{min}(l,m)$.  If $z$ is a $n \times p$ matrix, let $\mathbf{z}$ (in bold) denote
\[
\left( \begin{array}{c} z_1 \\ \vdots \\ z_{l-\kappa} \\ z_{l} \\ \vdots \\ z_{n-\kappa} \end{array} \right).
\]
In other words, $\mathbf{z}$ is equal to $z$, but with the following $2\kappa$ rows removed: $l-\kappa+1$, $l-\kappa+2$, ..., $l$ and $n-\kappa+1$, $n-\kappa+2$, ..., $n$.  Note that the definition of $\mathbf{z}$ depends on $\kappa$, even though this is not reflected explicitly in the notation. (We do not write, for example, $\mathbf{z}(\kappa)$.  This is to avoid notational clutter.)  The motivation for this ``bold'' notation is that we can use $\mathbf{z}$ as a training set.  
The remaining $2\kappa$ rows of $z$ can be used as a test set.  

Next, define the function $a(z)$ as follows:
\begin{equation}
a(z) \equiv \frac{1}{2\kappa} \left\{\sum_{i=l-\kappa+1}^{l} \hat{f}(z_i; \mathbf{z}) + \sum_{i=n-\kappa+1}^{n} \left[1 - \hat{f}(z_i; \mathbf{z})\right] \right\}.
\end{equation}
Here we use only $\mathbf{z}$ as the training set, so $n_0 = n-2\kappa$.  The remaining $2\kappa$ observations are the test set.  We count how many of the test-set observations are correctly classified, and divide by $2\kappa$.  Thus, $a(z)$ may be interpreted as the out-of-sample classification accuracy rate for one specific partition of $z$ into a training set and test set.

We may now define $s_{\mathrm{out}}(z)$ as:
\begin{equation}
s_{\mathrm{out}}(z) = \frac{1}{l!m!}\sum_{i,j}a\left(\Pi_i^{(X)} \Pi_j^{(Y)} z\right) \label{oosrate}
\end{equation}
where $\Pi_1^{(X)}$, $\Pi_2^{(X)}$, ..., $\Pi_{l!}^{(X)}$ denotes some ordering of the $l!$ permutation matrices that permute only the first $l$ rows of $z$, and $\Pi_1^{(Y)}$, $\Pi_2^{(Y)}$, ..., $\Pi_{m!}^{(Y)}$ denotes some ordering of the $m!$ permutation matrices that permute only the final $m$ rows of $z$.  In other words, $\{\Pi_i^{(X)}\}$ is the set of all $n\times n$ permutation matrices whose most lower-right $m\times m$ submatrix is equal to $I_{m\times m}$, and $\{\Pi_i^{(Y)}\}$ is the set of all $n\times n$ permutation matrices whose most upper-left $l\times l$ submatrix is equal to $I_{l\times l}$.  (Note that $\Pi_i^{(X)}$ and $\Pi_j^{(Y)}$ commute.)  Equation \ref{oosrate} is our definition of the \textit{out-of-sample classification accuracy rate}.  We will drop the subscript ``out'' on $s_{\mathrm{out}}(z)$ when it is clear from context.

\subsection{Consistency} \label{consist}

The CPT is a consistent test if (1) we use the out-of-sample classification accuracy rate, and (2) the classification function $\hat{f}$ has at least some predictive power to discriminate $\mathcal{F}$ from $\mathcal{G}$.  The exact sense in which we mean ``$\hat{f}$ has at least some predictive power to discriminate $\mathcal{F}$ from $\mathcal{G}$'' is specified in Definition \ref{kdgpredictive}.

\begin{definition} \label{kdgpredictive}
Let $Z$, $\kappa$, and $\mathbf{Z}$ be defined as above.  Let $\tilde{X} \sim \mathcal{F}$ and $\tilde{Y} \sim \mathcal{G}$ be $1 \times p$ random vectors, and assume that $\tilde{X}$ and $\tilde{Y}$ are independent of $Z$ and of each other.  We say that a function $\hat{f}: \mathbb{R}^p\times\mathbb{R}^{(n-2\kappa) \times p} \mapsto \{0, 1\}$ is $(\kappa, \delta, \gamma)$-predictive under $\mathcal{F}$ and $\mathcal{G}$ if and only if both of the following are true: 
$$
\mathbb{P}\left\{\mathbb{P}\left[\hat{f}\left(\tilde{X}, \mathbf{Z}\right) = 1 \given \mathbf{Z}\right] > 0.5 + \delta\right\} > 1-\gamma
$$
and
$$
\mathbb{P}\left\{\mathbb{P}\left[\hat{f}\left(\tilde{Y}, \mathbf{Z}\right) = 0 \given \mathbf{Z}\right] > 0.5 + \delta\right\} > 1-\gamma.
$$
\end{definition}
In other words, if we use $\mathbf{Z}$ as a training set, then with probability at least $1 - \gamma$ the function $\hat{f}$ will be able to correctly classify new, independent observations at least somewhat better than a coin flip.  Under this assumption, if $\gamma$ is sufficiently small and if $\kappa$ is sufficiently large, it follows that with high probability the test statistic $S$ will be at least some finite amount larger that 0.5.  More precisely:
\begin{proposition} \label{Sdistro}
Assume that $\mathcal{F} \ne \mathcal{G}$ and that $\hat{f}$ is $(\kappa, \delta, \gamma)$-predictive under $\mathcal{F}$ and $\mathcal{G}$.  Then 
\begin{equation}
\mathbb{P}\left[S \le 0.5 + \delta/4 \right] < \frac{8\gamma + 4\mathrm{exp}\left(-\kappa\delta^2\right)}{\delta}.
\end{equation}
\end{proposition}
\begin{proof}
See Appendix \ref{proofs}.
\end{proof}
Moreover, if $\kappa$ is large, then most of the the values of $S^{(i)}$ concentrate right around 0.5.
\begin{proposition} \label{sibound}
Let $\xi$ be some real number such that $0 < \xi < 0.5$.  Then 
\begin{equation}
\frac{\#\{i : S^{(i)} > 0.5 + \xi\}}{n!} < \frac{1}{\xi}\left(\frac{1+\sqrt{\pi}}{2\sqrt{2}}\right) \sqrt{\frac{1}{\kappa}}.
\end{equation}
\end{proposition}
\begin{proof}
See Appendix \ref{proofs}.
\end{proof}

Combining Propositions \ref{Sdistro} and \ref{sibound}, and recalling the definition of $P$, we see that the power of the CPT goes to 1 as as $n \to \infty$ as long as $\kappa \to \infty$ and $\gamma \to 0$ and $\delta \to \delta_0 > 0.5$.  Note that slightly stronger statements are possible using the results in Appendix \ref{proofs}, but we do not pursue them here.

\subsection{Comments} \label{techcomments}

To summarize, when constructing the test statistic we must make two main choices: (1) what classifier to use, and (2) what accuracy measure to use.  Neither decision affects the validity of the test (that is guaranteed by Proposition \ref{validity}) but our choices affect the power of the test, and also the computational complexity.

In practice, the most important choice is usually the classifier (see below).  The better the classifier can distinguish $\mathcal{F}$ from $\mathcal{G}$, the more powerful the test.  This is both a feature and a bug.  On one hand, a researcher may have some intuition about what type of classifier might best fit her data (e.g.\ a linear vs.\ non-linear classifier), and thus ``customize'' the CPT to her particular application.  We feel that this is a major strength of the method.  On the other hand, since the choice is arbitrary, it could easily lead to data snooping.  We therefore suggest, as a default, that researchers run the CPT once with logistic regression and once with random forests, and report the results of both.  If it is felt that a third classifier is more appropriate, we suggest reporting its result as well, in addition to the first two.  Of course, when the CPT is merely being used as a diagnostic tool to discover covariate imbalance, data snooping may not be a serious concern --- if there are serious imbalances, we would like to find them, even if it requires a little searching.     

The choice of accuracy measure seems to be much less important than the choice of classifier in practice.  See Figure \ref{fig: roc-inout} in the appendix, which compares the in-sample vs out-of-sample CPT on simulated data.  Little, if any, difference can be seen.  Our focus here on the out-of-sample classification rate is primarily theoretical; it is more difficult (and requires further assumptions) to prove consistency of the in-sample CPT.  To see why, consider a K-nearest neighbors classifier, with K = 1.  This classifier may be able to discriminate $\mathcal{F}$ from $\mathcal{G}$ in the sense described above (Section \ref{consist}), but the in-sample CPT will have 0 power.  Assuming the $Z_i$ are all distinct, the in-sample classification accuracy rate will always be 1, over all permutations, and thus the CPT will never reject. 

Another theoretical detail is that in our definition of the out-of-sample classification accuracy rate we force the test set to have an equal number of observations from treatment and control.  The idea here is that, to the extent that the classifier approximates an ideal Bayes classifier, it should approximate a Bayes classifier that has a 50/50 prior on the class label.  If the prior is not uniform on the class label, and especially if there is a large imbalance, it is possible that the Bayes classifier would always classify every observation to a single class.  In such cases, the classification accuracy rate would be constant over all permutations, and the CPT would have 0 power.  In practice, this implies that some caution may be required when applying the CPT to datasets with a large imbalance in the number of observations from treatment and control.  In such cases, it may be preferable to use the out-of-sample classification accuracy rate (instead of in-sample), and to ensure the classifier effectively places a uniform prior on the class label.

%%%%%%%%%%%%%%%%%%%%%%%%%%%%%%%%%%%%%%%%
%%%%%%%%%%%%%%  Discussion  %%%%%%%%%%%%
%%%%%%%%%%%%%%%%%%%%%%%%%%%%%%%%%%%%%%%%

\section{Discussion}  \label{discussion}

The CPT reformulates the problem of testing whether a binary treatment was assigned at random as a test for equality of multivariate distributions. The test combines classification methods with Fisherian permutation inference.  We illustrate the power of the method relative to existing procedures using Monte-Carlo simulations as well as four real data examples. We hope the CPT will illustrate the gains of using machine learning tools for the construction of powerful new test statistics, and Fisherian inference for conducting hypothesis testing and inference.

The paper emphasizes the importance of the joint distribution rather than the marginal distributions when testing for equality of multivariate distributions. The CPT is \emph{not} a substitute for standard methods such as a balance table that tests for differences in the means of each pre-treatment characteristic separately. The CPT is targeted to complement a balance table and provide a summary measure of the covariates' imbalance.

The CPT can be easily generalized.  Furthermore, although we focus in this paper on binary treatments, a similar method could be implemented for continuous treatments by replacing the classification algorithm with some form of regression, and replacing the classification accuracy rate with some other goodness of fit measure. This flexibility, combined with exact finite sample inference, allows researchers to verify random assignment to treatment in a variety of situations. The four empirical applications aim to illustrate the applicability of the method to different situations that rise in applied research.

\clearpage
\singlespacing
\bibliographystyle{aer}
\bibliography{references}
\doublespacing
\clearpage

\clearpage

\appendix

\section{Appendix: Supplementary Figures and Tables}

\begin{figure}[ht]
\centering
\caption{ROC curves comparing the performance of the in-sample and out-of-sample variants of the CPT on simulated data}
\label{fig: roc-inout}
\includegraphics[scale=0.4]{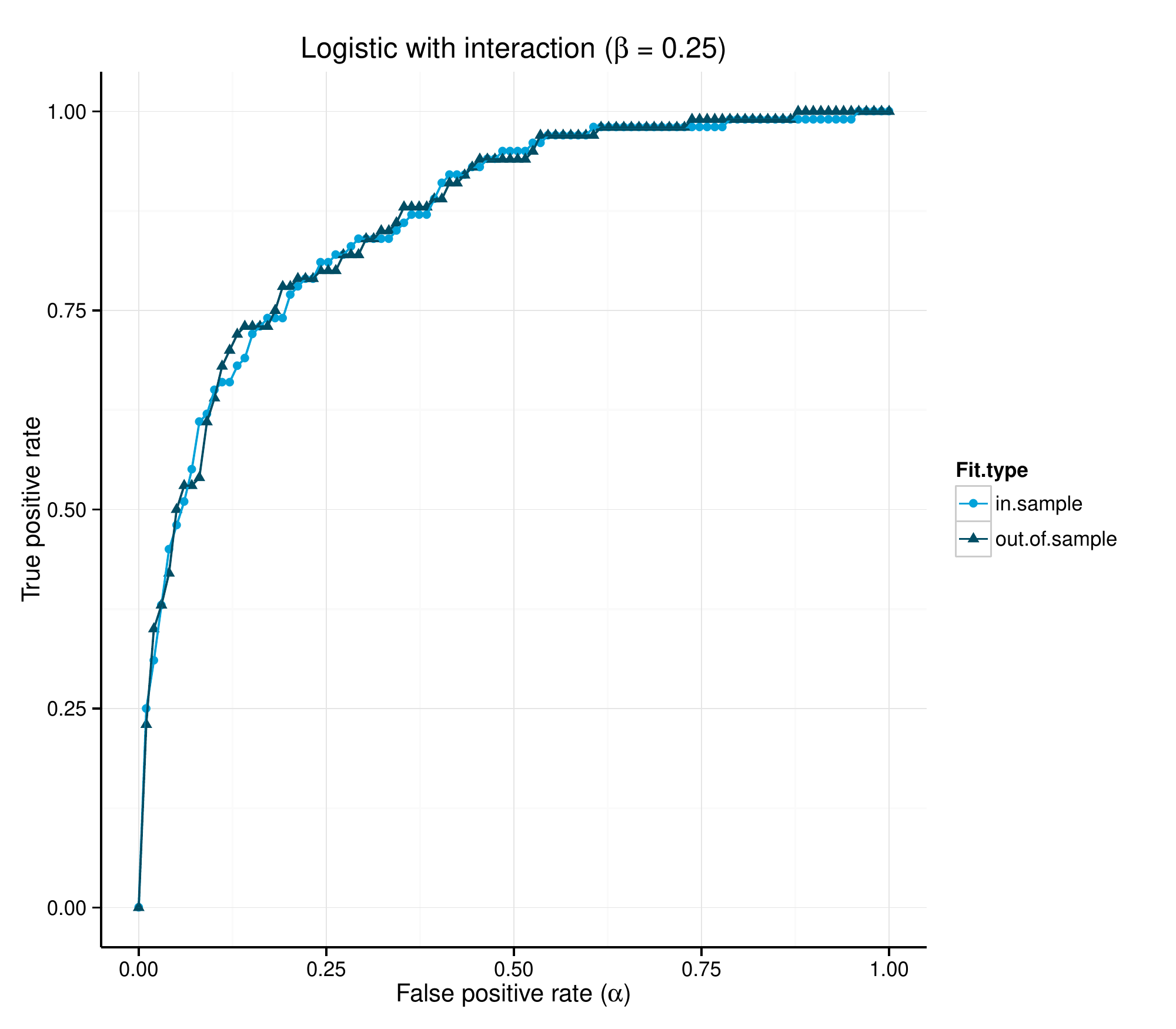} 
\includegraphics[scale=0.4]{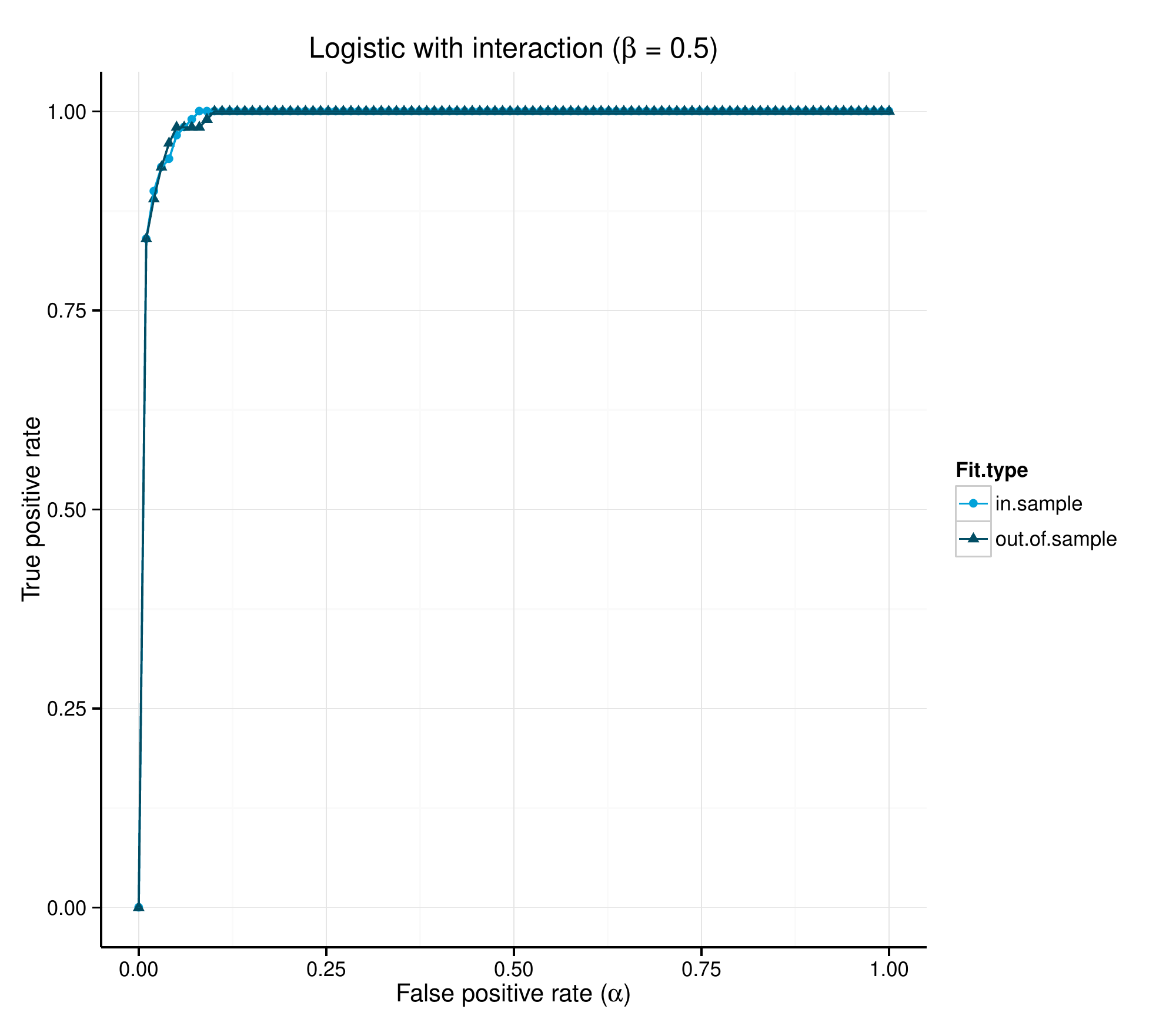} 
\begin{minipage}{16cm}
\emph{Notes:  Simulations were conducted as described in Section \ref{simulations}.  Only results for logistic2 were computed, since calculating the out-of-sample classification accuracy rate of random forests is computationally demanding.}
\end{minipage}
\end{figure}

\begin{figure}[ht]
\caption{Covariate balance table}
\label{fig: MPs balance table }
\centering
\includegraphics[scale=1.5]{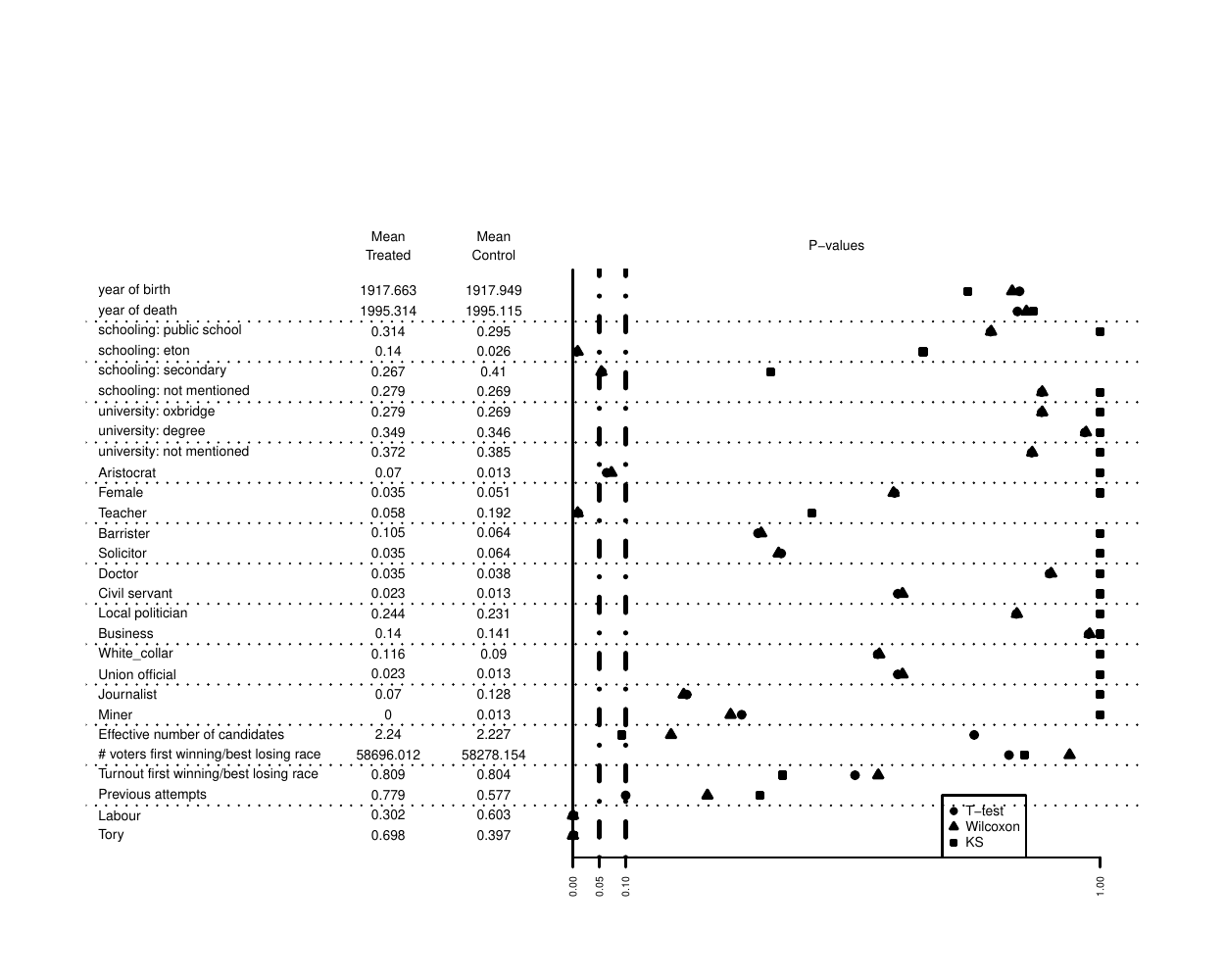}
\begin{minipage}{16cm}
\emph{Notes:  
Table 4 in EH shows the main estimates of the treatment effect. The estimates use a window of 164 to 223 observations around the winning threshold. In this figure we limited the sample to a window containing 164 observations.}    
\end{minipage}
\end{figure} 

\begin{figure}[ht]
\centering
\caption{The distribution of the winning margin by party identity \newline}
\label{fig: Null distribution CPT MPs}
\includegraphics[scale=.75]{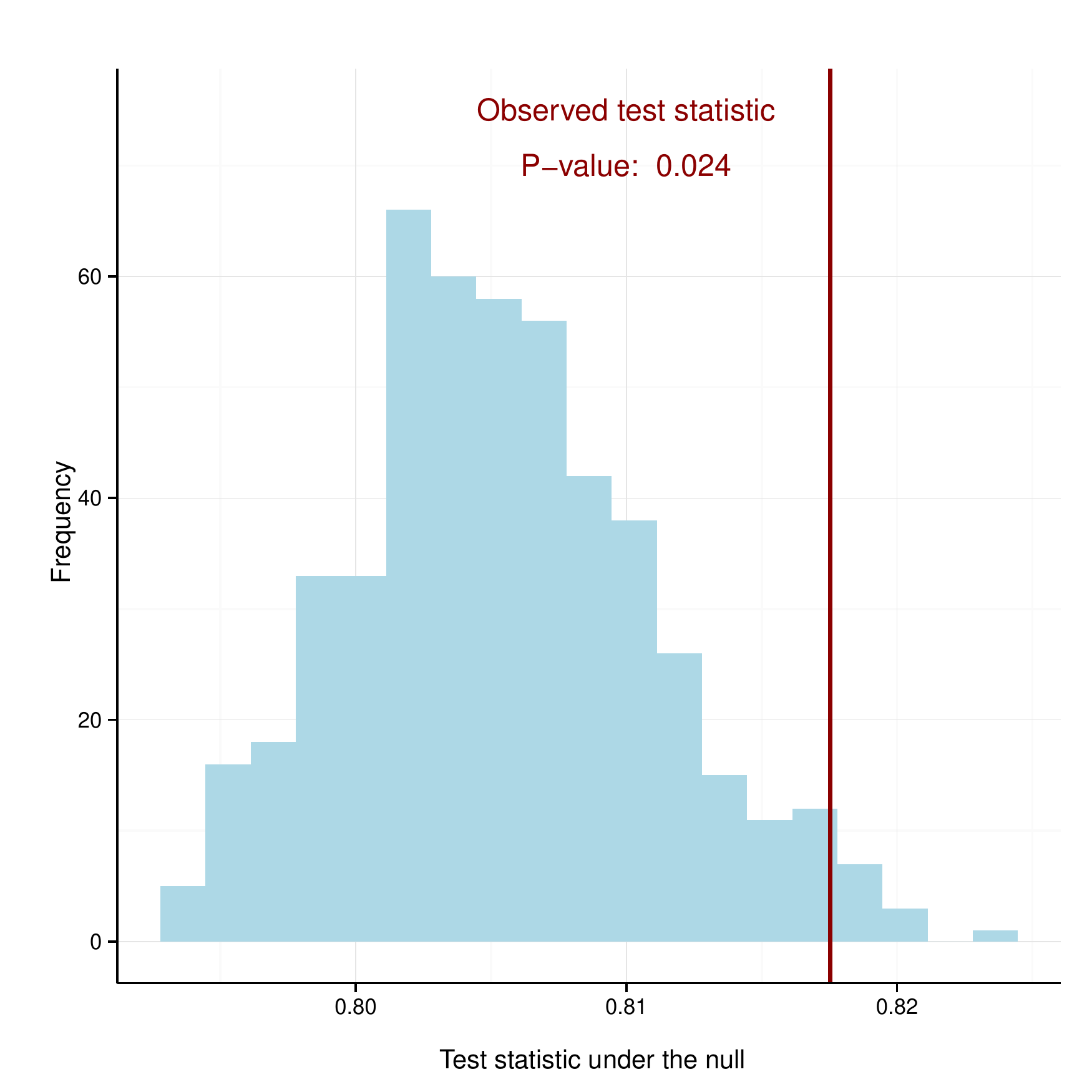}
\begin{minipage}{13cm}
\footnotesize
\emph{Notes: 
We used random forest as the classifier. Logistic regression with all two-way interactions will have more parameters than observations and therefore cannot be implemented.
Table 4 in EH shows the main estimates of the treatment effect. The estimates use a window of 164 to 223 observations around the winning threshold. In this figure we limited the sample to a window containing 164 observations.
}
\end{minipage}
\end{figure}

\begin{figure}[ht]
\caption{The distribution of the Likelihood Ratio Test P-value using all two-way interactions}
\label{fig: judges Type-I error rate LRT}
\centering
\includegraphics[scale=1]{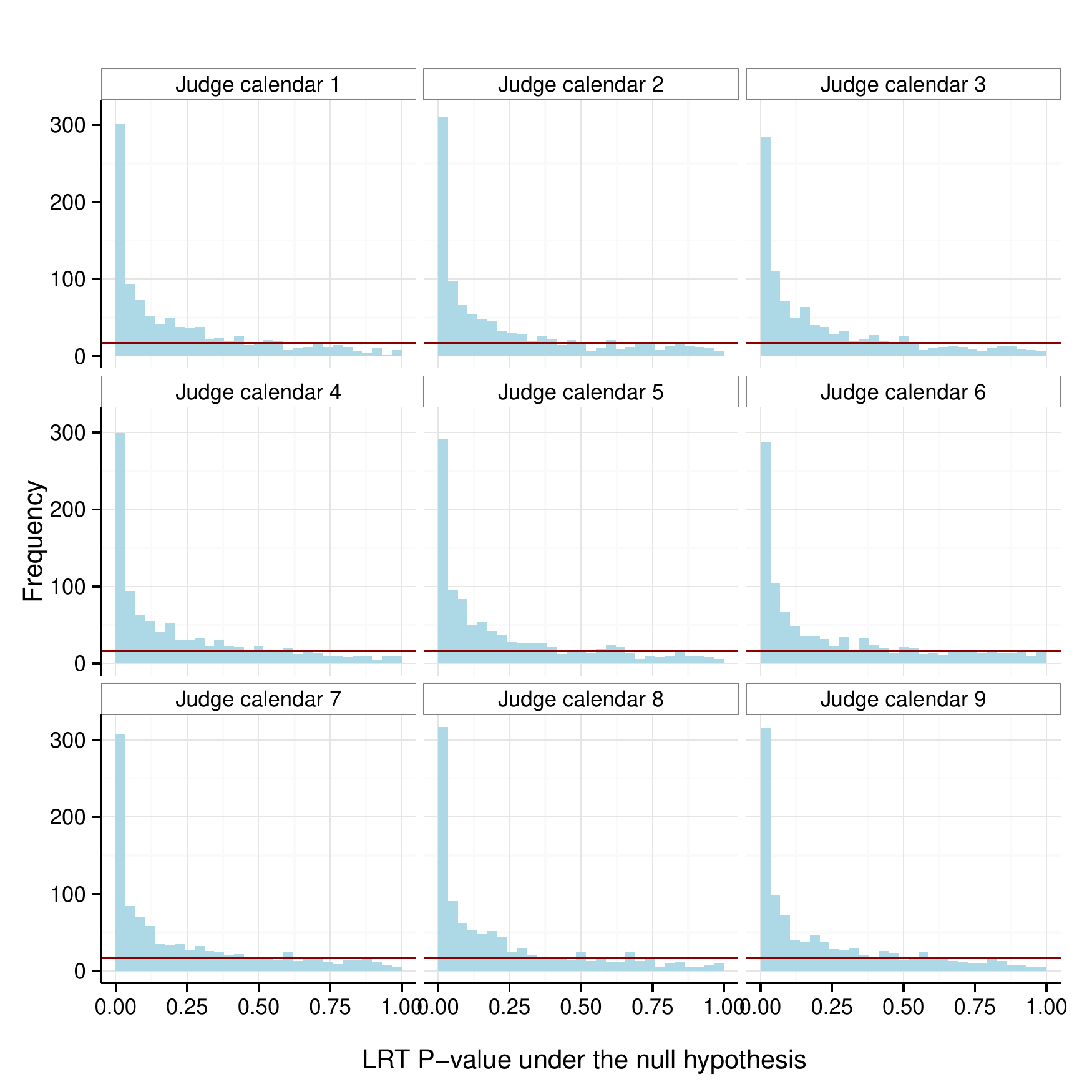}
\begin{minipage}{16cm}
\emph{Notes: Each plot shows the distribution of the LRT P-value under the null hypothesis of random assignment. The red line shows the uniform $[0,1]$ distribution that is expected under the null.}    
\end{minipage}
\end{figure} 

\clearpage
% Fri Oct  9 23:44:44 2015
\begin{table}[ht]
\centering
\caption{Balance table within the RDD window} 
\label{tab: MPs balance window}
\begin{threeparttable}
\begin{tabular}{rrrrrr}
  \hline
 & Ave. Treat & Ave. control & T-test & Wilcoxon & KS \\ 
  \hline
year of birth & 1917.66 & 1917.95 & 0.85 & 0.83 & 0.75 \\ 
  year of death & 1995.31 & 1995.12 & 0.84 & 0.86 & 0.87 \\ 
  schooling: public school & 0.31 & 0.29 & 0.79 & 0.79 & 1.00 \\ 
  schooling: eton & 0.14 & 0.03 & 0.01 & 0.01 & 0.66 \\ 
  schooling: secondary & 0.27 & 0.41 & 0.06 & 0.05 & 0.38 \\ 
  schooling: not mentioned & 0.28 & 0.27 & 0.89 & 0.89 & 1.00 \\ 
  university: oxbridge & 0.28 & 0.27 & 0.89 & 0.89 & 1.00 \\ 
  university: degree & 0.35 & 0.35 & 0.97 & 0.97 & 1.00 \\ 
  university: not mentioned & 0.37 & 0.39 & 0.87 & 0.87 & 1.00 \\ 
  Aristocrat & 0.07 & 0.01 & 0.06 & 0.07 & 1.00 \\ 
  Female & 0.04 & 0.05 & 0.61 & 0.61 & 1.00 \\ 
  Teacher & 0.06 & 0.19 & 0.01 & 0.01 & 0.45 \\ 
  Barrister & 0.10 & 0.06 & 0.35 & 0.36 & 1.00 \\ 
  Solicitor & 0.04 & 0.06 & 0.40 & 0.39 & 1.00 \\ 
  Doctor & 0.04 & 0.04 & 0.90 & 0.91 & 1.00 \\ 
  Civil servant & 0.02 & 0.01 & 0.62 & 0.62 & 1.00 \\ 
  Local politician & 0.24 & 0.23 & 0.84 & 0.84 & 1.00 \\ 
  Business & 0.14 & 0.14 & 0.98 & 0.98 & 1.00 \\ 
  White\_collar & 0.12 & 0.09 & 0.58 & 0.58 & 1.00 \\ 
  Union official & 0.02 & 0.01 & 0.62 & 0.62 & 1.00 \\ 
  Journalist & 0.07 & 0.13 & 0.22 & 0.21 & 1.00 \\ 
  Miner & 0.00 & 0.01 & 0.32 & 0.30 & 1.00 \\ 
  Effective number of candidates & 2.24 & 2.23 & 0.76 & 0.19 & 0.09 \\ 
  \# voters first winning/best losing race & 58696.01 & 58278.15 & 0.83 & 0.94 & 0.86 \\ 
  Turnout first winning/best losing race & 0.81 & 0.80 & 0.54 & 0.58 & 0.40 \\ 
  Previous attempts & 0.78 & 0.58 & 0.10 & 0.26 & 0.35 \\ 
  Labour & 0.30 & 0.60 & 0.00 & 0.00 & 0.00 \\ 
  Tory & 0.70 & 0.40 & 0.00 & 0.00 & 0.00 \\ 
   \hline
\end{tabular}
\begin{tablenotes}
\footnotesize
\emph{Notes: The covariate balance in a window around the cut-point that includes 164 observations. }
\end{tablenotes}
\end{threeparttable}
\end{table}

\begin{figure}[ht]
\begin{center}
\caption{The covariate balance before and after matching}
\label{fig: college balance}
\includegraphics[scale=0.5]{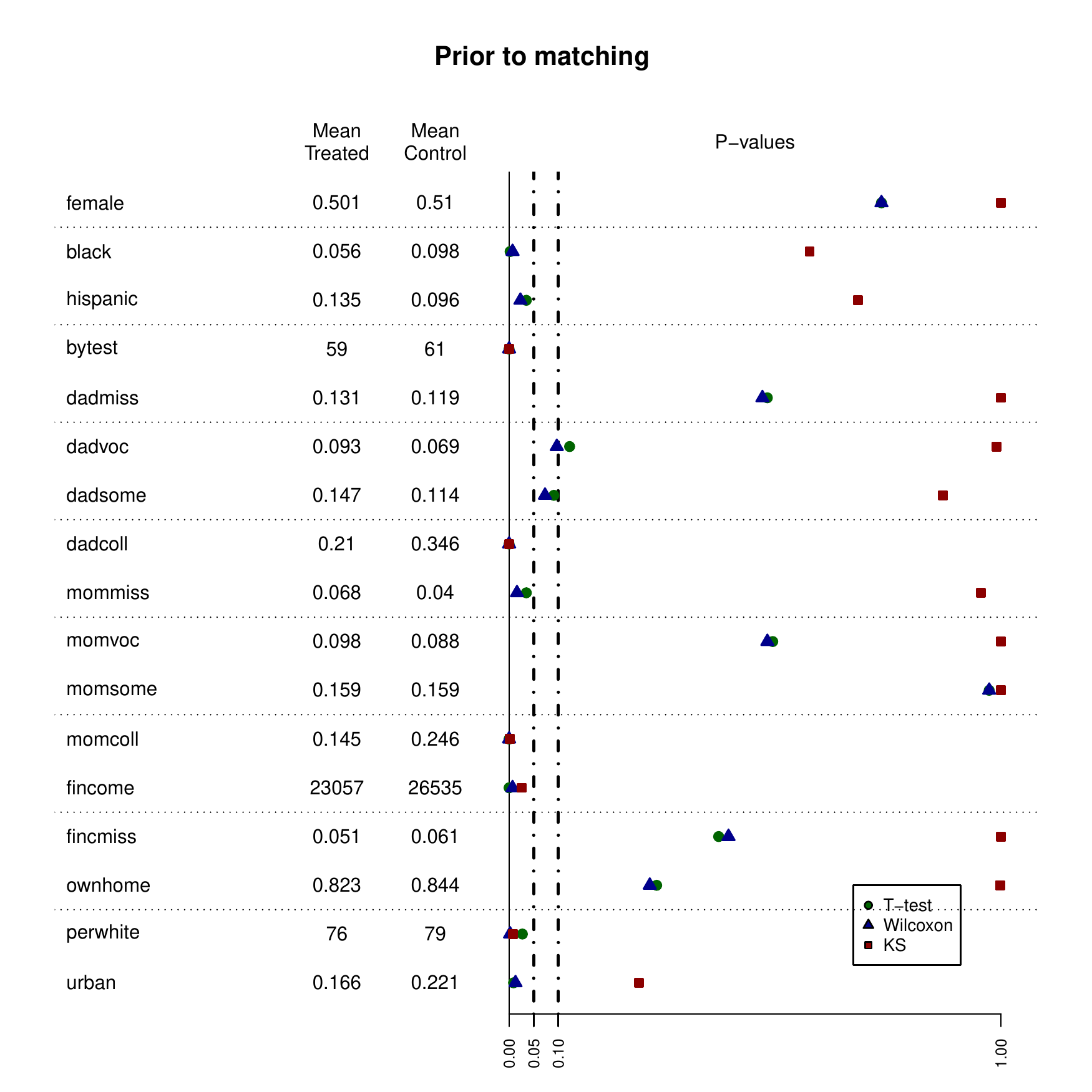}
\end{center}
\includegraphics[scale=0.5]{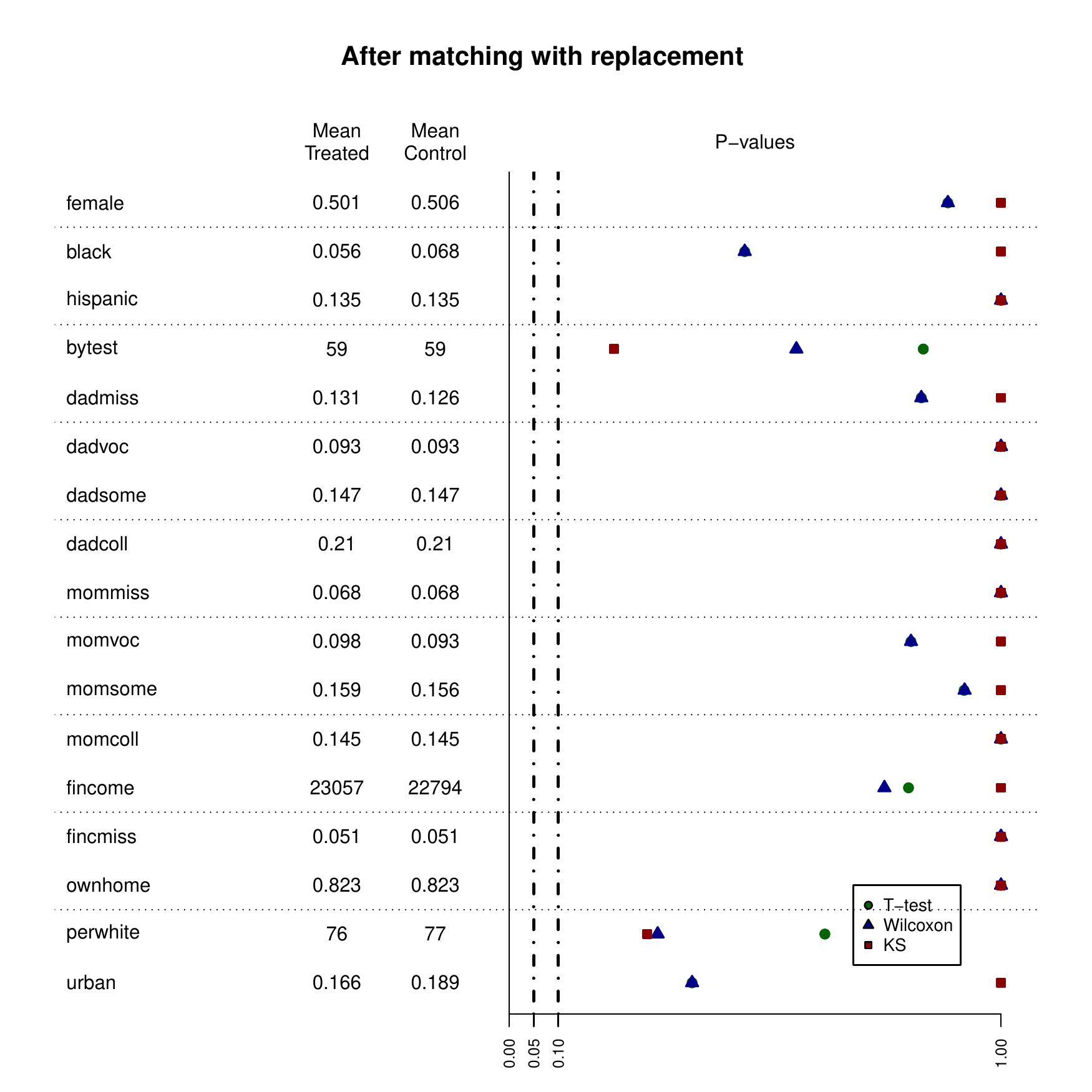}
\includegraphics[scale=0.5]{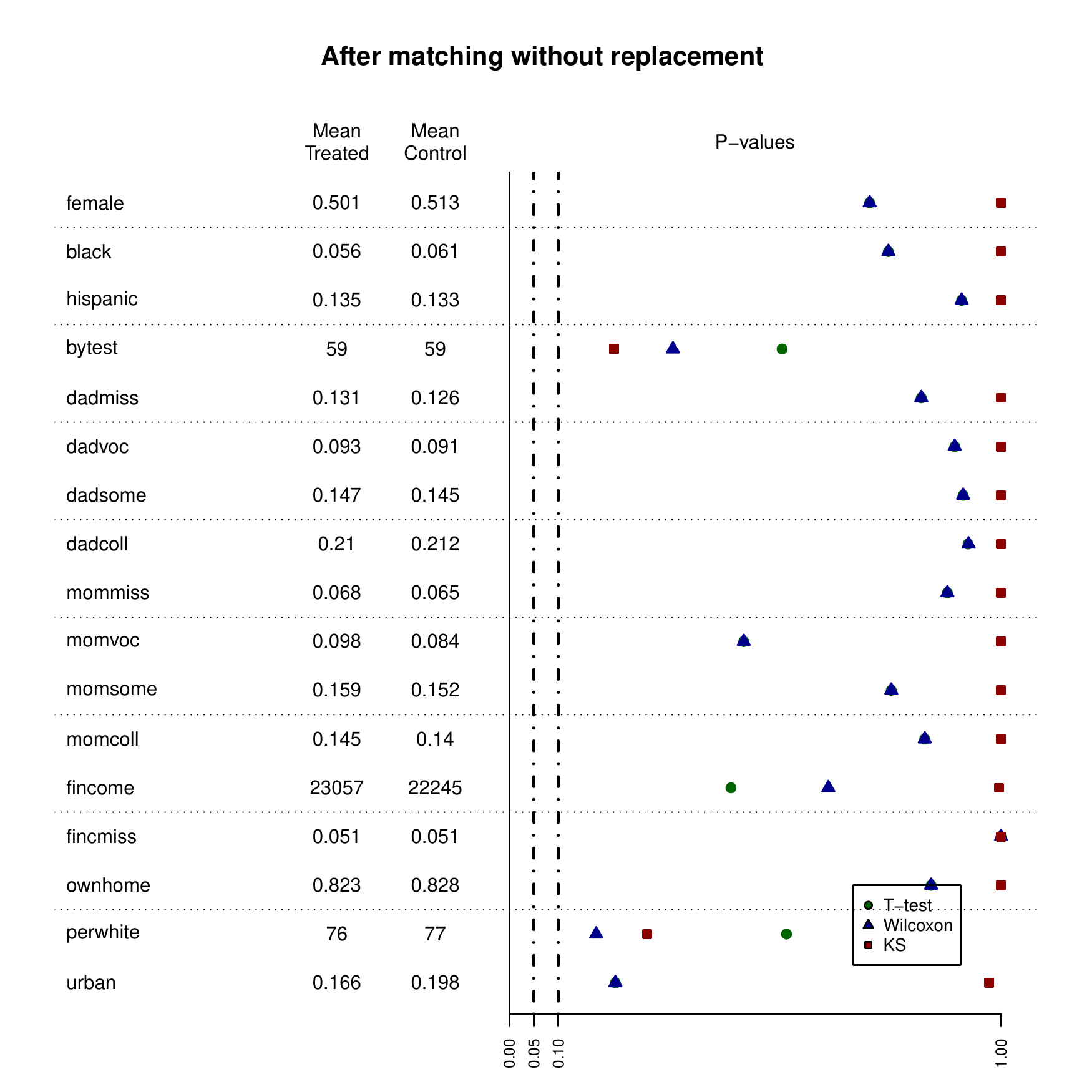}
\begin{minipage}{17cm}
\footnotesize
\emph{Notes:} The figures show the covariate balance in \cite{rouse1995} data before and after implementing a matching procedure to minimize distances on observable characteristics of students in two-year relative to students in four-year college.
\end{minipage}
\end{figure}

\FloatBarrier

%%%%%%%%%%%%%%%%%%%%%%%%%%%%%%%%%%%%%%%%%%%%%%%%%%%%%%%%%%%%%%%%%%%%%%%%%%%%%%%%%%%%%%%%%%%%%%%%%%%%%%%%%%%%
%%%%%%%%%%%%%%%%%%%%%%%%%%%%%%%%%%%%%%%%%%%%%%%%%%%%%%%%%%%%%%%%%%%%%%%%%%%%%%%%%%%%%%%%%%%%%%%%%%%%%%%%%%%%
%%%%%%%%%%%%%%%%%%%%%%%%%%%% Data description from Green and Winik (2010) %%%%%%%%%%%%%%%%%%%%%%%%%%%%%%%%%%
%%%%%%%%%%%%%%%%%%%%%%%%%%%%%%%%%%%%%%%%%%%%%%%%%%%%%%%%%%%%%%%%%%%%%%%%%%%%%%%%%%%%%%%%%%%%%%%%%%%%%%%%%%%%
%%%%%%%%%%%%%%%%%%%%%%%%%%%%%%%%%%%%%%%%%%%%%%%%%%%%%%%%%%%%%%%%%%%%%%%%%%%%%%%%%%%%%%%%%%%%%%%%%%%%%%%%%%%%

\section{Appendix: List of observed defendant characteristics, \cite{green2010}} \label{appendix: data green and winik (2010)}
                                                                        
\begin{verbatim}
"Indicator: defendant female"                                                     
"Indicator: defendant not black"                                                  
"Defendant's age at arrest"                                                       

"Indicator: defendant arrested prior to arrest in sampled case"                   
"Indicator: defendant arrested on felony charge prior to arrest in sampled case"  
"Indicator: defendant arrested on drug charge prior to arrest in sample case"     
"Indicator: defendant arrested on felony drug charge prior to arrest in sampled"
"Indicator: defendant convicted prior to arrest in sampled case"                  
"Indicator: defendant convicted on felony charge prior to arrest in sampled case" 
"Indicator: defendant convicted on drug charge prior to arrest in sampled case"   
"Indicator: defendant convicted on felony drug charge prior to arrest in sampled"

"Indicator: marijuana was drug involved in crime"                                 
"Indicator: powder cocaine was drug involved in crime"                            
"Indicator: crack cocaine was drug involved in crime"                             
"Indicator: heroin was drug involved in crime"                                    
"Indicator: PCP was drug involved in crime"                                       
"Indicator: another drug was involved in crime"                                   
"Indicator: defendant charged with felony possession with intent to distribute"   
"Indicator: defendant charged with felony distribution"   
\end{verbatim}

%%%%%%%%%%%%%%%%%%%%%%%%%%%%%%%%%%%%%%%%%%%%%%%%%%%%%%%%%%%%%%%%%%%%%%%%%%%%%%%%%%%%%%%%%%%%%%%%%%%%%%%%%%%%
%%%%%%%%%%%%%%%%%%%%%%%%%%%%%%%%%%%%%%%%%%%%%%%%%%%%%%%%%%%%%%%%%%%%%%%%%%%%%%%%%%%%%%%%%%%%%%%%%%%%%%%%%%%%
%%%%%%%%%%%%%%%%%%%%%%%%%%%% Proofs %%%%%%%%%%%%%%%%%%%%%%%%%%%%%%%%%%%%%%%%%%%%%%%%%%%%%%%%%%%%%%%%%%%%%%%%
%%%%%%%%%%%%%%%%%%%%%%%%%%%%%%%%%%%%%%%%%%%%%%%%%%%%%%%%%%%%%%%%%%%%%%%%%%%%%%%%%%%%%%%%%%%%%%%%%%%%%%%%%%%%
%%%%%%%%%%%%%%%%%%%%%%%%%%%%%%%%%%%%%%%%%%%%%%%%%%%%%%%%%%%%%%%%%%%%%%%%%%%%%%%%%%%%%%%%%%%%%%%%%%%%%%%%%%%%

\section{Appendix: Proofs} \label{proofs}
\setcounter{proposition}{0}
\setcounter{lemma}{0}
\setcounter{corollary}{0}
\setcounter{definition}{0}

This appendix is a (mostly) self-contained collection of technical results.  We prove the propositions presented in the main text, along with some others.  Many of the propositions presented below build off of one another, and are numbered accordingly.  As a consequence, the numbering of the propositions here does not correspond with the numbering of the propositions in the main text.  For quick reference, here is the correspondence:
\begin{itemize}
\item Proposition \ref{validity} in the main text $\mapsto$ Corollary \ref{validity3} below.
\item Proposition \ref{Sdistro} in the main text $\mapsto$ Corollary \ref{corosdistrosm} below.
\item Proposition \ref{sibound} in the main text $\mapsto$ Proposition \ref{siboundsm} below.
\end{itemize}  
For convenience, we also give here a quick reference of some notation from Section \ref{technical} of the main text: 
\begin{itemize}
\item $X$ is a $l\times p$ matrix whose rows $X_1$, ..., $X_l$ are IID from $\mathcal{F}$.
\item $Y$ is a $m \times p$ matrix whose rows $Y_1$, ..., $Y_m$ are IID from $\mathcal{G}$.
\item $ Z \equiv \left(\begin{array}{c} X \\ Y \end{array} \right) $
\item $s:\mathbb{R}^{n \times p}\mapsto \mathbb{R}$ is a fixed but otherwise arbitrary measurable function.
\item $\Pi_1$, ..., $\Pi_{n!}$ denotes some ordering of the $n!$ permutation matrices of dimension $n \times n$.  We assume that $\Pi_1 = I$, but the ordering may otherwise be arbitrary.
\item $S^{(i)} \equiv s(\Pi_i Z)$ %for $1 \le i \le n!$  
\item $\displaystyle P \equiv \frac{\#\left\{i : S \le S^{(i)} \right\}}{n!}$
\end{itemize}
In addition, we adopt the following (new) notation:
\begin{itemize}
\item $r(z) \equiv \#\left\{i : s\left(z\right) \le s\left(\Pi_i z\right) \right\}$ \\ (Note that the variable $z$ does not have any special meaning of its own; it is simply used here to define the function $r$ in terms of the function $s$.)
\item $R^{(i)} \equiv r(\Pi_i Z) = \#\left\{j : s\left(\Pi_i Z\right) \le s\left(\Pi_j \Pi_i Z\right) \right\}$
\end{itemize}
Finally, note the following equality:
\begin{itemize}
\item $\displaystyle P = \frac{R^{(1)}}{n!}$
\end{itemize}

%%%%%%%%%%%%%%%%%%%%%%%%%%%%%%%%%%%%%%%%%%%%%%%%%%%%%%%%%%%%%%%%%%%%%%%%%%%%%%%%%%%%%%%%%%%%%%%%%%%%%%%%%%%%
%%%%%%%%%%%%%%%%%%%%%%%%%%%% Validity %%%%%%%%%%%%%%%%%%%%%%%%%%%%%%%%%%%%%%%%%%%%%%%%%%%%%%%%%%%%%%%%%%%%%%
%%%%%%%%%%%%%%%%%%%%%%%%%%%%%%%%%%%%%%%%%%%%%%%%%%%%%%%%%%%%%%%%%%%%%%%%%%%%%%%%%%%%%%%%%%%%%%%%%%%%%%%%%%%%

\subsection{Validity of the CPT}

In this section we show that the CPT is a valid test (that it controls the type-I error rate).  We begin with Proposition \ref{simplevalidity}, which shows that the CPT is valid in the special case that the $S^{(i)}$ are all distinct.  With this added assumption, the proof is straight-forward.

\begin{proposition} \label{simplevalidity}
Assume that $\mathcal{F} = \mathcal{G}$ and that with probability 1, the values of the $S^{(i)}$ are all distinct.  Then for any real number $\alpha$ such that $0 \le \alpha \le 1$, it follows that $\mathbb{P}\left(P \le \alpha \right) \le \alpha$.  
\end{proposition}
\begin{proof}
Let $i$ and $j$ be integers such that $1 \le i \le n!$ and $1 \le j \le n!$.  Let $E_{ij}$ denote the event $\left\{R^{(i)} = j \right\}$.  In the following paragraphs we will show that for any value of $j$, the events $E_{1j}$, $E_{2j}$, ..., $E_{n!j}$ are disjoint and have have equal probabilities, and that the union of these events has probability 1.  From this it follows that $\mathbb{P}(E_{ij}) = 1/n!$ for all $i, j$.  This implies that $\mathbb{P}\left(R^{(i)} \le j\right) = \frac{j}{n!}$ for all $i$, which implies in particular that $\mathbb{P}\left(R^{(1)} \le j\right) = \frac{j}{n!}$, which further implies that $\mathbb{P}\left(P \le \frac{j}{n!}\right) = \frac{j}{n!}$.  The desired result then follows immediately.

To see that that the events $E_{1j}$, $E_{2j}$, ..., $E_{n!j}$ are all disjoint, note that 
\begin{align}
R^{(i)} = R^{(i')} &\implies \#\left\{k : s\left(\Pi_{i}Z\right) \le s\left(\Pi_k \Pi_{i}Z\right) \right\} = \#\left\{k : s\left(\Pi_{i'}Z\right) \le s\left(\Pi_k \Pi_{i'}Z\right) \right\}  \label{distinctline1}\\
&\implies \#\left\{k : s\left(\Pi_{i}Z\right) \le s\left(\Pi_k Z\right) \right\} = \#\left\{k : s\left(\Pi_{i'}Z\right) \le s\left(\Pi_k Z\right) \right\} \\
&\implies s\left(\Pi_{i}Z\right) = s\left(\Pi_{i'}Z\right) \\
&\implies i = i'  \label{distinctline4}
\end{align}
and thus $i \ne i' \implies R^{(i)} \ne R^{(i')}$.

To see that the events $E_{1j}$, $E_{2j}$, ..., $E_{n!j}$ all have equal probabilities, first note that because $\mathcal{F} = \mathcal{G}$, the rows of $Z$ are IID, and therefore exchangeable.  Thus for any $i$ we have that $\Pi_i Z$ is equal in distribution to $Z$.  From this it follows that the $R^{(i)}$ are all equal in distribution (recall $R^{(i)} \equiv r(\Pi_i Z)$).  This implies in particular that for any $j$ such that $1 \le j \le n!$ the events
\begin{equation}
\left\{R^{(1)} = j \right\}, \left\{R^{(2)} = j \right\}, ..., \left\{R^{(n!)} = j \right\}
\end{equation}
all have equal probability.  

Finally, note that the union of the events $E_{1j}$, $E_{2j}$, ..., $E_{n!j}$ has probability 1, because $R^{(i)}$ must be an integer between 1 and $n!$ (inclusive).
\end{proof}

Our next goal is to remove the restriction that the $S^{(i)}$ must all be distinct.  If the $S^{(i)}$ are not distinct, the proof of Proposition \ref{simplevalidity} breaks down, because the $R^{(i)}$ will not be distinct either, and thus the events $E_{1k}$, $E_{2k}$, ..., $E_{n!k}$ will not be disjoint.  To solve this problem, we will break ties at random.
 
We first introduce some additional notation.  
\begin{itemize}
\item Let $\Upsilon_1, ..., \Upsilon_{n!}$ be IID Uniform[0,1] random variables.
\item Let $\tilde{R}^{(i)} \equiv \#\left\{j : s\left(\Pi_i Z\right) < s\left(\Pi_j \Pi_i Z\right) \right\} + \#\left\{j : s\left(\Pi_i Z\right) = s\left(\Pi_j \Pi_i Z\right) \mbox{ and } \Upsilon_i < \Upsilon_j \right\}$
\item Let $\displaystyle \tilde{P} \equiv \frac{\tilde{R}^{(i)}}{n!}$ 
\end{itemize}

We can now show
\begin{proposition} \label{validity2}
Assume that $\mathcal{F} = \mathcal{G}$.  Then for any real number $\alpha$ such that $0 \le \alpha \le 1$, it follows that $\mathbb{P}\left(\tilde{P} \le \alpha \right) \le \alpha$.  
\end{proposition}
\begin{proof}
The $\tilde{R}^{(i)}$ are clearly all distinct with probability 1, and also identically distributed.  We may therefore simply replace $R^{(i)}$ with $\tilde{R}^{(i)}$ and make the appropriate changes to lines (\ref{distinctline1})--(\ref{distinctline4}) in the proof of Proposition \ref{simplevalidity}.
\end{proof}

Given Proposition \ref{validity2} we can also show
\begin{corollary} \label{validity3}
Assume that $\mathcal{F} = \mathcal{G}$.  Then for any real number $\alpha$ such that $0 \le \alpha \le 1$, it follows that $\mathbb{P}\left(P \le \alpha \right) \le \alpha$.  
\end{corollary}
\begin{proof}
Note that $P \ge \tilde{P}$, which implies that $\mathbb{P}\left(P \le \alpha \right) \le \mathbb{P}\left(\tilde{P} \le \alpha \right)$, and then cite Proposition \ref{validity2}.
\end{proof}

We would also like to generalize Proposition \ref{simplevalidity} to allow for the function $s$ to be random, in order to allow for randomized algorithms such as random forests.  This can be accomplished as follows.  Let $s^{(1)}$, $s^{(2)}$, ..., $s^{(n!)}$ be a sequence of random measurable functions mapping $\mathbb{R}^{n \times p}$ to $\mathbb{R}$.  Assume that the $s^{(i)}$ are mutually independent and identically distributed.  Then simply re-define $S^{(i)} \equiv s^{(i)}(\Pi_i Z)$.  Under this definition, the $S^{(i)}$ (and, by extension, the $R^{(i)}$) continue to be identically distributed.  If the $S^{(i)}$ are not all distinct, this can be handled as in Proposition \ref{validity2}.

%%%%%%%%%%%%%%%%%%%%%%%%%%%%%%%%%%%%%%%%%%%%%%%%%%%%%%%%%%%%%%%%%%%%%%%%%%%%%%%%%%%%%%%%%%%%%%%%%%%%%%%%%%%%
%%%%%%%%%%%%%%%%%%%%%%%%%%%% Power %%%%%%%%%%%%%%%%%%%%%%%%%%%%%%%%%%%%%%%%%%%%%%%%%%%%%%%%%%%%%%%%%%%%%%%%%
%%%%%%%%%%%%%%%%%%%%%%%%%%%%%%%%%%%%%%%%%%%%%%%%%%%%%%%%%%%%%%%%%%%%%%%%%%%%%%%%%%%%%%%%%%%%%%%%%%%%%%%%%%%%

\subsection{Power}

This section is composed of three subsections.  The first contains some notation, including both notation from the main text and new notation.  The second subsection contains results that lead to the proof of Proposition \ref{Sdistro} in the main text, and the third contains results that lead to the proof of Proposition \ref{sibound} in the main text. 

%%%%%%%%%%%%%%%%%%%%%%%%%%%% Notation %%%%%%%%%%%%%%%%%%%%%%%%%%%%%%%%%%%%%%%%%%%%%%%%%%%%%%%%%%%%%%%%%%%%%%
\subsubsection{Notation}

Recall Definition 1 from the main text.
\begin{definition}
Let $Z$, $\kappa$, and $\mathbf{Z}$ be defined as in the main text.  Let $\tilde{X} \sim \mathcal{F}$ and $\tilde{Y} \sim \mathcal{G}$ be $1 \times p$ random vectors, and assume that $\tilde{X}$ and $\tilde{Y}$ are independent of $Z$ and of each other.  We say that a function $\hat{f}: \mathbb{R}^p\times\mathbb{R}^{(n-2\kappa) \times p} \mapsto \{0, 1\}$ is $(\kappa, \delta, \gamma)$-predictive under $\mathcal{F}$ and $\mathcal{G}$ if and only if both of the following are true: 
$$
\mathbb{P}\left\{\mathbb{P}\left[\hat{f}\left(\tilde{X}, \mathbf{Z}\right) = 1 \given \mathbf{Z}\right] \ge 0.5 + \delta\right\} > 1-\gamma
$$
and
$$
\mathbb{P}\left\{\mathbb{P}\left[\hat{f}\left(\tilde{Y}, \mathbf{Z}\right) = 0 \given \mathbf{Z}\right] \ge 0.5 + \delta\right\} > 1-\gamma.
$$
\end{definition}
\noindent
Recall also the following notation from the main text.
\begin{itemize}
\item $\Pi_1^{(X)}$, $\Pi_2^{(X)}$, ..., $\Pi_{l!}^{(X)}$ denotes some ordering of the subset of $n\times n$ permutation matrices that permute only the first $l$ rows
\item $\Pi_1^{(Y)}$, $\Pi_2^{(Y)}$, ..., $\Pi_{m!}^{(Y)}$ denotes some ordering of the subset of $n\times n$ permutation matrices that permute only the final $m$ rows
\item $\displaystyle a(z) \equiv \frac{1}{2\kappa} \left\{\sum_{i=l-\kappa+1}^{l} \hat{f}(z_i; \mathbf{z}) + \sum_{i=n-\kappa+1}^{n} \left[1 - \hat{f}(z_i; \mathbf{z})\right] \right\}$
\item $\displaystyle s(z) = \frac{1}{l!m!}\sum_{i,j}a\left(\Pi_i^{(X)} \Pi_j^{(Y)} z\right)$
\end{itemize}
In addition:
\begin{itemize}
\item Let $\Pi_{ijk} \equiv \Pi_i^{(X)}\Pi_j^{(Y)}\Pi_k$
\item Let $Z^{(i)} \equiv \Pi_{i}Z$
\item Let $Z^{(ijk)} \equiv \Pi_{ijk}Z$
\item Let $\hat{F}_j \equiv \hat{f}\left(Z_j; \mathbf{Z}\right)$
\item Let $\hat{F}^{(i)}_j \equiv \hat{f}\left(Z^{(i)}_j; \mathbf{Z^{(i)}}\right)$
\item Let $A \equiv a\left(Z\right)$
\item Let $A^{(ijk)} \equiv a\left(Z^{(ijk)}\right)$
\end{itemize}
Note in particular:
\begin{itemize}
\item $\displaystyle S^{(k)} = s(\Pi_kZ) = \frac{1}{l!m!}\sum_{i,j}a\left(\Pi^{(ijk)} Z\right) = \frac{1}{l!m!}\sum_{i,j}A^{(ijk)}$
\item $\displaystyle S = S^{(1)} = \frac{1}{l!m!}\sum_{i,j}A^{(ij1)}$
\end{itemize}
In order to more easily index the training set and test set, we also define:
\begin{itemize}
\item Let $\mathbf{T}_X \equiv \{1,2,...,l-\kappa\}$ 
\item Let $T_X \equiv \{l-\kappa+1,l-\kappa+2,...,l\}$
\item Let $\mathbf{T}_Y \equiv \{l+1,l+2,...,n-\kappa\}$
\item Let $T_Y \equiv \{n-\kappa+1, n-\kappa+2,...,n\}$
\item Let $\mathbf{T} \equiv \mathbf{T}_X \cup \mathbf{T}_Y$
\item Let $T \equiv T_X \cup T_Y$
\item Let $\mathbf{t}$ be an arbitrary element of $\mathbf{T}$, let $\mathbf{t}_X$ be an arbitrary element of $\mathbf{T}_X$, let $t$ be an arbitrary element of $T$, etc.
\end{itemize} 

%%%%%%%%%%%%%%%%%%%%%%%%%%%% Bounding S %%%%%%%%%%%%%%%%%%%%%%%%%%%%%%%%%%%%%%%%%%%%%%%%%%%%%%%%%%%%%%%%%%%%
\subsubsection{Bounding $\mathbb{P}\left[S \le 0.5 + \delta/4 \right]$}

%%% Lemma
\begin{lemma} \label{markovlemma}
Let $U_1$, $U_2$, ..., $U_N$ be real-valued non-negative random variables with finite expectations, and let $c$, $d$, and $e$ be positive real numbers.  
Assume $\mathbb{P}(U_i \le c) < e$ for all $i$.  Then $\mathbb{P}(\bar{U} \le c - d) < \frac{ec}{d}$.
\end{lemma}
\begin{proof}
Let 
\begin{equation}
V_i \equiv \begin{cases} 0 & \mbox{if } U_i < c \\ c & \mbox{if } c \le U_i \end{cases} 
\end{equation}
and note that $V_i \le U_i$ for all $i$, and thus $\bar{V} \le \bar{U}$.  Note also that $V_i \le c$ for all $i$, and thus $\bar{V} \le c$.  Let $\mu_i \equiv \mathbb{E}\left(V_i\right)$ and $\mu \equiv \mathbb{E}\left(\bar{V}\right)$.   Note that $\mu_i > (1-e)c$ for all $i$, and thus $\mu > (1-e)c$.  

Let $W \equiv c - \bar{V}$.  Then $W \ge 0$ and $\mathbb{E}(W) = c - \mu$.  Thus by Markov's inequality
\begin{equation}
\mathbb{P}\left(W \ge d \right) \le \frac{c-\mu}{d}
\end{equation}
which implies that
\begin{align}
\mathbb{P}\left(\bar{V} \le c - d \right) &\le \frac{c-\mu}{d} \\
&< \frac{c-(1-e)c}{d} \\
&= \frac{ec}{d}.
\end{align}
Since $\bar{V} \le \bar{U}$, it follows that $\mathbb{P}\left(\bar{U} \le c - d \right) < \frac{ec}{d}$.
\end{proof}

%%% Proposition
\begin{proposition} \label{sdistrosm}
Assume that $\mathcal{F} \ne \mathcal{G}$ and that $\hat{f}$ is $(\kappa, \delta, \gamma)$-predictive under $\mathcal{F}$ and $\mathcal{G}$, with $\delta < 0.5$.  Let $\epsilon$ and $\zeta$ be positive real numbers such that $\epsilon < \zeta < \delta$. Then 
\begin{equation}
\mathbb{P}\left[S \le 0.5 + \epsilon \right] < \frac{2\gamma + \mathrm{exp}\left[-4\kappa(\delta-\zeta)^2\right]}{\zeta - \epsilon}.
\end{equation}
\end{proposition}
\begin{proof}
Let 
\begin{equation}
p_1 \equiv \mathbb{P}\left[\hat{f}\left(\tilde{X}, \mathbf{Z}\right) = 1 \given \mathbf{Z}\right]
\end{equation}
and
\begin{equation}
p_2 \equiv \mathbb{P}\left[\hat{f}\left(\tilde{Y}, \mathbf{Z}\right) = 0 \given \mathbf{Z}\right].
\end{equation}
From the fact that $\hat{f}$ is $(\kappa, \delta, \gamma)$-predictive under $\mathcal{F}$ and $\mathcal{G}$, it immediately follows that $\mathbb{P}(p_1 < 0.5 + \delta) < \gamma$ and $\mathbb{P}(p_2 < 0.5 + \delta) < \gamma$.  Let $E_1$ denote the event $\{p_1 < 0.5 + \delta\} \cup \{p_2 < 0.5 + \delta\}$ and let $E_2$ denote the event $\{p_1 \ge 0.5 + \delta\} \cap \{p_2 \ge 0.5 + \delta\}$.  Then
\begin{align}
\mathbb{P}\left[A \le 0.5 + \zeta \right] &= 
  \mathbb{P}\left[A  \le 0.5 + \zeta \given E_1 \right]\mathbb{P}\left[ E_1 \right] + \mathbb{P}\left[A  \le 0.5 + \zeta \given E_2 \right]\mathbb{P}\left[ E_2 \right] \\
&\le  \mathbb{P}\left[ E_1 \right] + \mathbb{P}\left[A  \le 0.5 + \zeta \given E_2 \right] \\
 &< 2\gamma + \mathbb{P}\left[A  \le 0.5 + \zeta \given E_2 \right]. \label{prebinomial}
\end{align}

Note that $\hat{F}_{t_X} = \hat{f}\left(Z_{t_X}; \mathbf{Z}\right)$ is an indicator for whether observation $t_X$ in the test set is correctly classified, and also that, conditional on $p_1$, $\hat{F}_{t_X}$ is Bernoulli($p_1$).  Similarly, $\hat{F}_{\mathbf{t}_Y}$ is an indicator for whether observation $t_Y$ in the test set is \textit{incorrectly} classified, and that, conditional on $p_2$, $\hat{F}_{\mathbf{t}_Y}$ is Bernoulli($1-p_2$).  Note also that, conditional on $p_1$ and $p_2$, all of the $\hat{F}_{t_X}$ and $\hat{F}_{t_Y}$ are mutually independent.  Thus, conditional on $p_1$ and $p_2$, the distribution of 
\begin{equation}
2\kappa A = \sum_{t_X} \hat{f}(Z_{t_X}; \mathbf{Z}) + \sum_{t_Y} \left[1 - \hat{f}(Z_{t_Y}; \mathbf{Z})\right]
\end{equation}
 is equal to the distribution of the sum of two independent Binomial random variables, one with parameters $(\kappa, p_1)$ and the other with parameters $(\kappa, p_2)$.  
We therefore deduce that 
\begin{equation}
\mathbb{P}\left[2\kappa A  \le \kappa + 2\kappa\zeta \given \{p_1 \ge 0.5 + \delta\} \cap \{p_2 \ge 0.5 + \delta\} \right] \le \mathbb{P}\left[2\kappa A \le \kappa + 2\kappa\zeta \given p_1 = p_2 = 0.5 + \delta \right]
\end{equation}
which allows us to simplify (\ref{prebinomial}) to
\begin{align}
\mathbb{P}\left[A \le 0.5 + \zeta \right] &< 2\gamma + \mathbb{P}\left[A \le 0.5 + \zeta \given p_1 = p_2 = 0.5 + \delta \right]  \label{prehoeffding} .
\end{align}
Now, conditional on $p_1 = p_2 = 0.5 + \delta$, the distribution of $2\kappa A$ is Binomial($2\kappa,0.5+\delta$).  Using Hoeffding's bound for the binomial distribution, we simplify (\ref{prehoeffding}) to
\begin{align}
\mathbb{P}\left[A \le 0.5 + \zeta \right] &< 2\gamma + \mathrm{exp}\left[-4\kappa(\delta-\zeta)^2\right]. 
\end{align}
We are now in a position to bound $\mathbb{P}\left[S < 0.5 + \epsilon \right]$.  Observe that 
\begin{equation}
S = \frac{1}{l!m!}\sum_{ij} A^{ij1}
\end{equation}
and note also that $A^{ij1} \,{\buildrel d \over =}\, A$ for all $i,j$.  We may therefore apply Lemma \ref{markovlemma} with 
\begin{align}
c &= 0.5 + \zeta \\
e &= 2\gamma + \mathrm{exp}\left[-4\kappa(\delta-\zeta)^2\right] \\
d &= \zeta - \epsilon
\end{align}
to obtain
\begin{align}
\mathbb{P}\left[S \le 0.5 + \epsilon \right] &< \frac{\left\{2\gamma + \mathrm{exp}\left[-4\kappa(\delta-\zeta)^2\right]\right\}(0.5 + \zeta)}{\zeta - \epsilon}\\
&< \frac{2\gamma + \mathrm{exp}\left[-4\kappa(\delta-\zeta)^2\right]}{\zeta - \epsilon}.
\end{align}
\end{proof}

%%% Corollary
\begin{corollary} \label{corosdistrosm}
Assume that $\mathcal{F} \ne \mathcal{G}$ and that $\hat{f}$ is $(\kappa, \delta, \gamma)$-predictive under $\mathcal{F}$ and $\mathcal{G}$, with $\delta < 0.5$.  Then 
\begin{equation}
\mathbb{P}\left[S \le 0.5 + \delta/4 \right] < \frac{8\gamma + 4\mathrm{exp}\left(-\kappa\delta^2\right)}{\delta}.
\end{equation}
\end{corollary}
\begin{proof}
Let $\epsilon = \delta/4$ and $\zeta = \delta/2$ in Proposition \ref{sdistrosm}.
\end{proof}

%%%%%%%%%%%%%%%%%%%%%%%%%%%% Bounding Si %%%%%%%%%%%%%%%%%%%%%%%%%%%%%%%%%%%%%%%%%%%%%%%%%%%%%%%%%%%%%%%%%%%

\subsubsection{Bounding $\#\{i : S^{(i)} \ge 0.5 + \epsilon\}$}

%%% Notation
We begin with some additional notation:
\begin{itemize}
\item Let $A^{(i)} \equiv a\left(Z^{(i)}\right)$
\item Let $\displaystyle \phi^{(i)}_X \equiv \sum_{t_X} \hat{F}^{(i)}_{t_X}$
\item Let $\displaystyle \phi^{(i)}_Y \equiv \sum_{t_Y} \hat{F}^{(i)}_{t_Y}$
\item Let $\phi^{(i)} = \phi^{(i)}_X + \phi^{(i)}_Y$
\end{itemize}
Note that with these definitions
\begin{align}
A^{(i)} &= \frac{1}{2\kappa} \left(\phi^{(i)}_X + \kappa - \phi^{(i)}_Y \right) \\
 &= 0.5 + \left(\phi^{(i)}_X - \phi^{(i)}_Y \right)/(2\kappa) \label{simpleai} 
\end{align}

%%% Bound Aijk
\begin{proposition} \label{aijkbound}  Let $\epsilon$ be some positive real number.  Then
\begin{equation}
\frac{\#\{(i,j,k) : A^{(ijk)} \ge 0.5 + \epsilon\}}{l!m!n!} \le \mathrm{exp}\left(-2\kappa\epsilon^2\right).
\end{equation}
\end{proposition}
\begin{proof}
First note that the multiset of all $\Pi_{ijk}$ is simply equal to the set of all $\Pi_i$, but with each element having multiplicity $l!m!$.  Thus
\begin{equation}
\frac{\#\{(i,j,k) : A^{(ijk)} \ge 0.5 + \epsilon\}}{l!m!n!} = \frac{\#\{i : A^{(i)} \ge 0.5 + \epsilon\}}{n!},  
\end{equation}
so we only need to count the number of $A^{(i)}$ that are greater than or equal to $0.5 + \epsilon$.  

Let us now partition the $n!$ permutation matrices $\{\Pi_i\}$ into $n!/(2\kappa)!$ disjoint subsets, which we denote $H_1$, $H_2$, ..., $H_{n!/(2\kappa)!}$, each of which contains exactly $(2\kappa)!$ matrices.  We choose the partitions such that, within each partition, the rows corresponding to the training set are fixed.  More formally, if we let $h_{ij}$ index of the elements of $H_i$, i.e.
\begin{equation}
H_{i} = \{\Pi_{h_{i1}}, ..., \Pi_{h_{i,(2\kappa)!}}\},
\end{equation}
and if we let $\left(\Pi_{h_{ij}}\right)_\mathbf{t}$ denote an arbitrary row $\mathbf{t} \in \mathbf{T}$ of matrix $\Pi_{h_{ij}}$, then our partitions are defined so that
\begin{equation}
 \left(\Pi_{h_{ij}}\right)_\mathbf{t} = \left(\Pi_{h_{ij'}}\right)_\mathbf{t}
\end{equation}
for all $i$, $j$, $j'$, and $\mathbf{t}$.  By way of contrast, note that the rows of the test set do vary.  If $j \ne j'$ and $t \in T$, then
\begin{equation}
 \left(\Pi_{h_{ij}}\right)_{t} \ne \left(\Pi_{h_{ij'}}\right)_{t}.
\end{equation}
To reduce notational clutter, in the discussion that follows we will let $H$ denote an arbitrary partition $H_i$, and we will let $h$ index an arbitrary element of $H$.   

Now, because the training set is fixed within the partition, the function $\hat{f}(\cdot, \mathbf{Z}^{(h)})$ is identical for all $h$, and in particular the sum
\begin{equation}
\phi^{(h)} = \sum_{t} \hat{f}\left(Z_{t}, \mathbf{Z}^{(h)}\right) 
\end{equation}
is the same for all $h$.  We therefore denote this quantity $\phi^{(H)}$.  Making use of (\ref{simpleai}), we see that
\begin{align}
A^{(h)} \ge 0.5 + \epsilon &\iff 0.5 + \left(\phi^{(h)}_X - \phi^{(h)}_Y \right)/(2\kappa) \ge 0.5 + \epsilon \\
&\iff \phi^{(h)}_X - \phi^{(h)}_Y  \ge 2\kappa\epsilon \\
&\iff 2\phi^{(h)}_X - \phi^{(H)}  \ge 2\kappa\epsilon \\
&\iff \phi^{(h)}_X \ge \kappa\epsilon  + \phi^{(H)}/2 
\end{align}
and thus
\begin{equation}
\frac{\#\{h : A^{(h)} \ge 0.5 + \epsilon\}}{(2\kappa)!} = \frac{\#\{h : \phi^{(h)}_X \ge \phi^{(H)}/2 + \kappa\epsilon \}}{(2\kappa)!}.  
\end{equation}
We are therefore interested in the proportion of $\phi^{(h)}_X$ that are greater than or equal to $C \equiv \phi^{(H)}/2 + \kappa\epsilon$.  

Observe that for any particular $h$, exactly $\phi^{(H)}$ of the 
$\hat{F}^{(h)}_t$ 
are equal to 1, and the remaining $2\kappa - \phi^{(H)}$ of the $\hat{F}^{(h)}_t$ are equal to 0.  The value of $\phi^{(h)}_X$ is equal to the number of ``1''s that have been allocated to the $\kappa$ rows of $T_X$.  Since the elements of $H$ include all possible shufflings within the test set of these ``1''s and ``0''s, the proportion 
\begin{equation}
\frac{\#\{h : \phi^{(h)}_X = k\}}{(2\kappa)!}  
\end{equation}
follows a hypergeometric distribution over $k$, with parameters $(2\kappa, \phi^{(H)}, \kappa)$.  ($2\kappa$ is the population size, $\phi^{(H)}$ is the number of ``successes'' within the population, and $\kappa$ is the sample size.)  

Therefore, using the results of \cite{chvatal1979tail}, 
\begin{align}
\frac{\#\{h : \phi^{(h)}_X \ge \phi^{(H)}/2 + \kappa\epsilon \}}{(2\kappa)!} &\le \mathrm{exp}\left\{-2\kappa\left[\left(\frac{\kappa-1}{2\kappa}\right)\phi^{(H)} + \epsilon \right]^2\right\} \\
&\le \mathrm{exp}\left(-2\kappa\epsilon^2\right). 
\end{align}
Thus, equivalently,
\begin{equation}
\frac{\#\{h : A^{(h)} \ge 0.5 + \epsilon\}}{(2\kappa)!} \le \mathrm{exp}\left(-2\kappa\epsilon^2\right).
\end{equation}
Since this inequality holds within each partition $H$, it also holds for all $A^{(i)}$, i.e.
\begin{equation}
\frac{\#\{i : A^{(i)} \ge 0.5 + \epsilon\}}{n!} \le \mathrm{exp}\left(-2\kappa\epsilon^2\right).
\end{equation}
\end{proof}

%%% Partition Lemma
The following lemma uses its own notation.
\begin{lemma} \label{partitionbound}
Let $U_{m\times n} = \{u_{ij} \}$ be a $m\times n$ matrix such that $u_{ij} \ge 0$ for all $i,j$.  Let $v_j \equiv \frac{1}{m}\sum_{i=1}^m u_{ij}$.  Let $G(t)$ be some function such that
\begin{equation}
\frac{\#\{(i,j) : u_{ij} > t \}}{mn} \le G(t)
\end{equation}
for all $t$.  Let $\epsilon$ be some real number such that $0 < \epsilon < 1$. Then
\begin{equation}
\frac{\#\left\{j : v_j > \int_0^{\infty} \mathrm{min}\left\{G(t)/\epsilon, 1 \right\} dt \right\}}{n} \le \epsilon
\end{equation}
\end{lemma}
\begin{proof}
Let $u_{(i)}$ denote the ``reverse order statistics'' of $\{u_{ij} \}$, i.e. 
\begin{equation}
u_{(1)} \ge u_{(2)} \ge u_{(3)} \ge ... \ge u_{(mn)} 
\end{equation}
and let $u^{(k)}$ denote the partial averages
\begin{equation}
u^{(k)} = \frac{1}{k}\sum_{i=1}^{k} u_{(i)}.
\end{equation}
Let 
\begin{equation}
F(t) = \frac{1}{mn}\sum_{i=1}^{mn} I(u_{(i)} \le t)
\end{equation}
and let
\begin{equation}
F^{(k)}(t) = \frac{1}{k}\sum_{i=1}^{k} I(u_{(i)} \le t)
\end{equation}
and note that
\begin{align}
u^{(k)} &= \int_0^{\infty} \left[1-F^{(k)}(t)\right] dt \\
&= \int_0^{\infty} \mathrm{min}\left\{\frac{mn}{k}\left[1-F(t)\right], 1 \right\} dt \\
&\le \int_0^{\infty} \mathrm{min}\left\{\frac{mn}{k}G(t), 1 \right\} dt.
\end{align}
Fix some value for $k$ and let $q = \left\lceil{\frac{k}{m}}\right\rceil$.  Note $k \le qm$.  Assume without loss of generality that $v_1 \ge v_2 \ge ... \ge v_n$.  Then
\begin{align}
\frac{1}{q} \sum_{j=1}^q v_j &= \frac{1}{mq} \sum_{i=1}^m\sum_{j=1}^q u_{ij}\\
&\le \frac{1}{mq} \sum_{i=1}^{mq} u_{(i)}\\
&\le u^{(k)}
\end{align}
and therefore there exists a $j \le q$ such that $v_j \le u^{(k)}$.  This implies that 
\begin{equation}
\#\{j : v_j > u^{(k)} \} \le q-1
\end{equation}
which further implies 
\begin{equation}
\frac{\#\left\{j : v_j > \int_0^{\infty} \mathrm{min}\left\{\frac{mn}{k}G(t), 1 \right\} dt \right\}}{n} \le \frac{q-1}{n}. \label{almostthere}
\end{equation}
But 
\begin{equation}
\frac{q-1}{n}  = \frac{\left\lceil{\frac{k}{m}}\right\rceil-1}{n} \le \frac{1}{n} \left\lfloor{\frac{k}{m}}\right\rfloor
\end{equation}
so plugging in to (\ref{almostthere}) gives
\begin{equation}
\frac{\#\left\{j : v_j > \int_0^{\infty} \mathrm{min}\left\{\frac{mn}{k}G(t), 1 \right\} dt \right\}}{n} \le \frac{1}{n} \left\lfloor{\frac{k}{m}}\right\rfloor
\end{equation}
from which it follows that
\begin{equation}
\frac{\#\left\{j : v_j > \int_0^{\infty} \mathrm{min}\left\{G(t)/\epsilon, 1 \right\} dt \right\}}{n} \le \epsilon
\end{equation}
for any $0 < \epsilon < 1$.
\end{proof}

\begin{proposition} \label{siboundsm}  
Let $\xi$ be some real number such that $0 < \xi < 1$.  Then 
\begin{equation}
\frac{\#\{k : S^{(k)} > 0.5 + \xi\}}{n!} < \frac{1}{\xi}\left(\frac{1+\sqrt{\pi}}{2\sqrt{2}}\right) \sqrt{\frac{1}{\kappa}}.
\end{equation}
\end{proposition}
\begin{proof}
Recall that
\begin{equation}
S^{(k)} = \frac{1}{l!m!}\sum_{i,j}A^{(ijk)}
\end{equation}
and recall from Proposition \ref{aijkbound} that for $\epsilon > 0$,
\begin{equation}
\frac{\#\{(i,j,k) : A^{(ijk)} \ge 0.5 + \epsilon\}}{l!m!n!} \le \mathrm{exp}\left(-2\kappa\epsilon^2\right).
\end{equation}
If we construct a $l!m!\times n!$ matrix $U$ such that the $k\thth$ column of $U$ contains the values of all $A^{(ijk)}$ (i.e. for all $(i,j)$, with $k$ held fixed), and set 
\begin{equation}
G(t) =  \begin{cases} 1 & \mbox{if } t < 0.5 \\ \mathrm{exp}\left[-2\kappa(t-0.5)^2\right] & \mbox{if } t \ge 0.5 \end{cases} 
\end{equation}
 then we may apply the results of Lemma \ref{partitionbound} to get
\begin{equation}
\frac{\#\left\{k : S^{(k)} > \int_0^{\infty}\mathrm{min}\left\{G(t)/\epsilon,1 \right\} dt\right\}}{n!} \le \epsilon.
\end{equation}
Assume $\epsilon < 1$.  Now let 
\begin{equation}
d \equiv \sqrt{\frac{-\mathrm{log}(\epsilon)}{2\kappa}}
\end{equation}
and note that
\begin{align}
\int_0^{\infty} \mathrm{min}\left\{G(t)/\epsilon, 1 \right\} dt &= 
\int_{0}^{0.5+d} \mathrm{min}\left\{G(t)/\epsilon, 1 \right\} dt +
\int_{0.5+d}^{\infty} \mathrm{min}\left\{G(t)/\epsilon, 1 \right\} dt \\
&= \int_{0}^{0.5+d} dt +
\int_{0.5+d}^{\infty} \frac{\mathrm{exp}\left[-2\kappa(t-0.5)^2\right]}{\epsilon} dt \\
&= 0.5 + d + \frac{1}{\epsilon}\int_{0.5+d}^{\infty} \mathrm{exp}\left[-2\kappa(t-0.5)^2\right] dt \\
&< 0.5 + d + \frac{1}{\epsilon}\int_{0.5}^{\infty} \mathrm{exp}\left[-2\kappa(t-0.5)^2\right] dt \\
&= 0.5 + d + \frac{1}{2 \epsilon} \sqrt{\frac{\pi}{2\kappa}}\\
&= 0.5 + \left(\sqrt{-\mathrm{log}(\epsilon)} + \frac{\sqrt{\pi}}{2 \epsilon}\right) \sqrt{\frac{1}{2\kappa}} \\ 
&< 0.5 + \left(\frac{1}{2\epsilon} + \frac{\sqrt{\pi}}{2 \epsilon}\right) \sqrt{\frac{1}{2\kappa}} \\ 
&= 0.5 + \frac{1}{\epsilon}\left(\frac{1+\sqrt{\pi}}{2\sqrt{2}}\right) \sqrt{\frac{1}{\kappa}} 
\end{align}
and thus
\begin{equation}
\frac{\#\left\{k : S^{(k)} >  0.5 +  \frac{1}{\epsilon}\left(\frac{1+\sqrt{\pi}}{2\sqrt{2}}\right) \sqrt{\frac{1}{\kappa}}
\right\}}{n!} \le \epsilon. \label{epsilonlessthanone}
\end{equation}
Note that (\ref{epsilonlessthanone}) also trivially holds if $\epsilon \ge 1$ and is therefore true for any $\epsilon > 0$.
If we now set $\epsilon$ to be
\begin{equation}
\epsilon = \frac{1}{\xi}\left(\frac{1+\sqrt{\pi}}{2\sqrt{2}}\right) \sqrt{\frac{1}{\kappa}} 
\end{equation}
we find
\begin{equation}
\frac{\#\left\{k : S^{(k)} >  0.5 + \xi \right\}}{n!} \le \frac{1}{\xi}\left(\frac{1+\sqrt{\pi}}{2\sqrt{2}}\right) \sqrt{\frac{1}{\kappa}}.
\end{equation}
\end{proof}

\end{document}